\documentclass[journal]{IEEEtran}

\usepackage{epsfig}
\usepackage[usenames]{color}
\usepackage{amsfonts}
\usepackage{amsmath}
\usepackage{amsthm}
\usepackage{amssymb}
\usepackage{algorithm}
\usepackage{algorithmic}
\usepackage{subfigure}
\usepackage{multirow}
\usepackage{graphicx}
\usepackage{hyperref}
\usepackage{url}
\usepackage{doi}

\newcommand{\abs}[1]{\lvert #1 \rvert}

\newcommand{\card}[1]{\abs{#1}}

\newcommand{\normsq}[1]{{\|{#1}\|}^2}
\newcommand{\ind}[1]{{\mathbb I}_{\{#1\}}}

\newcommand{\smfrac}[2]{{\textstyle{\frac{#1}{#2}}}}

\newcommand{\mU}{{\mathcal{U}}}
\newcommand{\mM}{{\mathcal{M}}}
\newcommand{\mP}{{\mathcal{P}}}
\newcommand{\mI}{{\mathcal{I}}}
\newcommand{\bp}{{\mathbf{p}}}

\newcommand{\bz}{{\mathbf{z}}}
\newcommand{\mD}{{\mathcal{D}}}
\newcommand{\mL}{{\mathcal{L}}}
\newcommand{\bPi}{{\boldsymbol{\Pi}}}
\newcommand{\changeclr}{\color{black}}
\newcommand{\normalclr}{\color{black}}

\newtheorem{proposition}{Proposition}

\newtheorem{remark}{Remark}

\renewcommand{\baselinestretch}{0.91}

\begin{document}


\title{Algorithms for Enhanced Inter Cell Interference Coordination (eICIC)
in LTE HetNets}

\author{Supratim Deb, Pantelis Monogioudis, Jerzy Miernik, James P. Seymour

\thanks{Manuscript received July~6, 2012; revised December~12, 2012; accepted
January~21, 2013.}

\thanks{S. Deb, P. Monogioudis, and J. Miernik are with
Wireless Chief Technology Office, Alcatel-Lucent USA.
e-mail: {\it first\_name.last\_name@alcatel-lucent.com}. The work was
done when J. P. Seymour was with Alcatel-Lucent.}
}

\maketitle

\begin{abstract}

The success of LTE Heterogeneous Networks (HetNets) with macro cells and pico cells
critically depends on efficient spectrum sharing between high-power macros and
low-power picos. Two important challenges in this context are, {(i)} determining the
amount of radio resources that macro cells should {\em offer} to pico cells, and {(ii)}
determining the association rules that decide which UEs should associate with
picos. In this paper, we develop a novel algorithm to solve these two coupled
problems in a joint manner. Our algorithm has provable guarantee, and furthermore, 
it accounts for network topology, traffic load, and macro-pico interference map. Our solution is
standard compliant and can be implemented using the notion of Almost Blank Subframes
(ABS) and Cell Selection Bias (CSB) proposed by LTE standards. We also show
extensive evaluations using RF plan from a real network and discuss SON based eICIC
implementation.

\end{abstract}


\vspace{-0.15in}
\section{Introduction}


Wireless data traffic has seen prolific growth in recent years due to new generation
of wireless gadgets (e.g., smartphones, tablets, machine-to-machine communications)
and also due to fundamental shift in traffic pattern from being data-centric to
video-centric. Addressing this rapid growth in wireless data calls for making
available radio spectrum as spectrally-efficient as possible. A key centerpiece is
making the radio spectrum efficient is LTE heterogeneous networks (LTE HetNet) or
{\em small cell} networks~\cite{lteA-overview}. In a HetNet architecture, in addition
to usual macro cells, wireless access is also provided through low-powered and
low-cost radio access nodes that have coverage radius around
10~m\--300~m\cite{eicic-intro}. Small cells in LTE networks is a general term used to
refer to Femto cells and Pico cells. \changeclr Femto cells are typically for indoor
use with a coverage radius of few tens of meters and its use is restricted to a
handful of users in {\em closed subscriber group}.  Pico cells have a coverage of
couple of hundreds of meters and pico cells are {\em open subscriber group} cells
with access permission to all subscribers of the operator. Picos are typically
deployed near malls, offices, business localities with dense mobile usage etc. Picos
are mostly deployed outdoor but there could be indoor deployments in large
establishments etc. \normalclr However, in LTE, since pico cells typically share the frequency
band as macro cells, the performance of a low-power pico access node could be
severely impaired by interference from a high power macro access node. Addressing
this interference management riddle is key to realize the true potential of a LTE
HetNet deployment and is the goal of this work. This work focuses on resource sharing
between macro cells and pico cells.  Note that macros and picos are both deployed in
a  planned manner by cellular operators. 


\begin{figure}[t]
\begin{center}
\includegraphics[height=1.1in,width=3.0in]{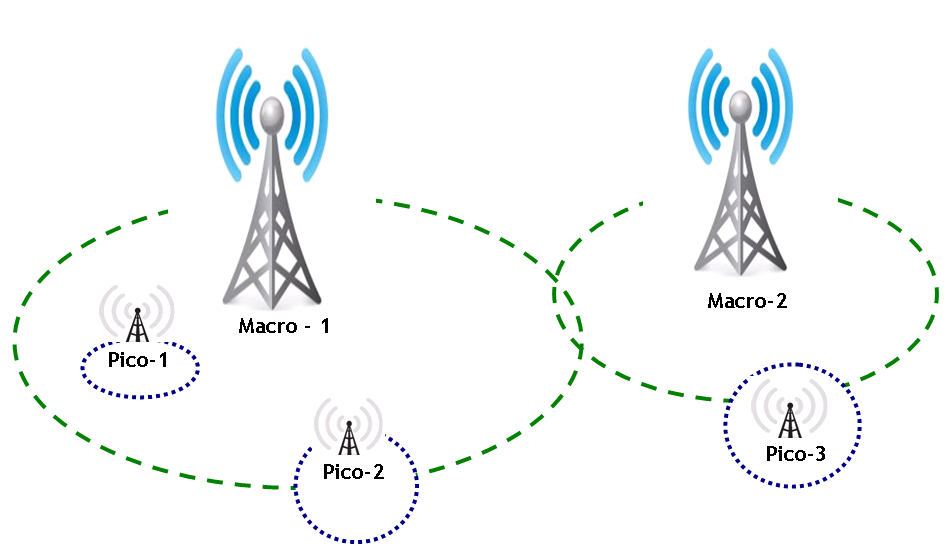}
\caption{\label{fig:hetnet}
A typical LTE HetNet architecture with Macro and Pico access nodes. Pico\--1
is used for throughput enhancement in a possible traffic hotspot location, Pico\--2
and Pico\--3 are used for improving edge throughput.}
\end{center}
\vspace{-0.20in}
\end{figure}

A typical HetNet with pico and macro~access nodes is shown in
Figure~\ref{fig:hetnet}.  The high-power macro network nodes are deployed for blanket
coverage of urban, suburban, or rural areas; whereas, the pico nodes with small RF
coverage areas aim to complement the macro network nodes for filling coverage holes
or enhancing throughput. 
There are two factors that could handicap the net capacity
of a pico access node in the downlink. Firstly, the downlink pico transmissions to its associated UEs
could be severely interfered by high power macro transmissions~\cite{survey-eicic}.
For e.g., in Figure~\ref{fig:hetnet}, downlink transmissions to UEs associated with
Pico\--1 could easily be interfered by downlink transmissions of Macro\--1.
Secondly, UEs, who are close to pico and could benefit from associating with a
pico access node, could actually end up associating with the macro access node due
to higher received signal strength from the high power macro access node \footnote{In
LTE networks, UEs associate typically with the cell with highest {\em Received
Signal Reference Power} (RSRP). RSRP is a measure of the received signal strength of
a cell at a UE and it is measured based on the strength of certain reference signals
that cells broadcast.}.  For e.g., UEs not too close to Pico\--3 but still
within the coverage area of Pico\--3 could end up associating with Macro\--2 because
of higher received signal strength from Macro\--2.  Indeed, this could leave the
pico underutilized and thus defeating the purpose of deploying that pico. 
\changeclr
Note that, it is the downlink interference at the pico UEs that needs additional protection from
the macros; the uplink interference at the picos can be mitigated using the same power
control principle in a macro only LTE network~\cite{ltefpc}.
\normalclr
Thus, for a pico cell based HetNet deployment to realize the promised theoretical gains, there
are two important questions that need to be answered:

\begin{enumerate}

\item  How should downlink radio resources be shared so that pico UEs are
guaranteed a {\em fair} share of throughput? Clearly, one needs to ensure that the
pico transmissions are not badly hit by interference from macros.

\item How to decide which UEs get associated with picos? 
Clearly, association based on highest signal
strength is inadequate to address this challenge.

\end{enumerate}

\changeclr

This paper provides answers to these two coupled questions.  \changeclr Realizing the
need to protect {\em downlink} pico transmissions by mitigating interference from
neighboring macro cells, 3GPP has proposed the notion of {\em enhanced inter cell
interference coordination} (eICIC) that provide means for macro and pico access nodes
to time-share the radio resources for downlink transmissions. \normalclr In simple
terms eICIC standards propose two techniques.  Firstly, each macro remains silent for
certain periods, termed Almost Blank Subframes (ABS periods), over which pico can
transmit at reduced interference. Secondly, the received signal strength based UE
association in LTE is allowed to be biased towards the pico by a suitable margin. The
details of how to set these ABS periods and how much to bias the association in favor
of picos are left unspecified. In this paper, we answer these questions. We design our
algorithms to meet the following goals: network-wide high performance, adaptability
to network settings like propagation map and network topology etc., and scalability.

\vspace{-0.15in}
\subsection{Our Contributions}

In this work, we make the following contributions:

\begin{enumerate}

\item {\em Framework for network dependent eICIC:} To the best of our knowledge, ours
is the first work to provide a formal framework for optimizing  Almost Blank
Subframes (ABS) and UE-association in every cell by accounting for cell specific
UE (load) locations, propagation map of each cell, macr-pico interference maps,
and network topology.  We also establish that computing the optimal solution with
respect to maximizing a network utility is computationally hard.

\item {\em Efficient eICIC Algorithms:} We next provide an efficient algorithm to
compute ABS and UE-associations (and corresponding CSB) in an LTE HetNet. Our
algorithm is provably within a constant factor of the optimal and scales linearly
with the number of cells.  Furthermore, our algorithm is amenable to distributed
implementation.

\changeclr
\item {\em Evaluation using Real RF Plan:}  We perform extensive evaluation of our
algorithm on a Radio-Frequency map from a real LTE deployment in New York City and demonstrate the
gains. The results show that, our algorithm performs within 90\% of the optimal for
realistic deployment scenarios, and, $5^{th}$ percentile of UE throughput in the
pico coverage area can improve up to more than 50\% compared to no eICIC; the
improvements can be $2\times$ for lower throughput percentiles.

\item {\em Practical Feasibility with SON:} Finally, we discuss the challenges of
implementing eICIC within {\em Self-Optimizing} (SON) framework and describe a
prototype along with the associated challenges.
\normalclr
\end{enumerate}

The rest of the paper is organized as follows. Section~\ref{sec:eicic} provides a
background on eICIC and describes
some important related work. In Section~\ref{sec:model}, we describe our network
model. Section~\ref{sec:son} states the problem and formally derives the
computational limits of the problem. Our main algorithm for jointly optimizing ABS
parameters for each cell and macro/pico association for each UE is provided in
Section~\ref{sec:algo1}\--\ref{sec:algo3}. Section~\ref{sec:csbabspattern} describes how a
given choice of UE association can be translated into cell selection bias parameters
and also how ABS numbers can be converted into ABS patterns.
Section~\ref{sec:eval} presents evaluations using RF plan from a real LTE deployment.
Finally, Section~\ref{sec:son} discusses SON based eICIC implementation along
with a prototype.

\vspace{-0.15in}
\section{Background: eICIC and Related Work}
\label{sec:eicic}

\subsection{Enhanced Inter Cell Interference Coordination (eICIC)}

The eICIC proposal in LTE standards serves two important purposes:
allow for time-sharing of spectrum resources (for downlink transmissions)
between macros and picos so as to
mitigate interference to pico in the downlink, and, allow for flexibility
in UE association so that picos are neither underutilized nor overloaded.
In eICIC, a macro eNodeB can inject silence periods in its
transmission schedule from time to time, so that interfering pico eNodeBs can use
those silence periods for downlink transmissions. Furthermore, to ensure that
sufficient number of UEs get associated with a pico, the eICIC mechanism allow
UEs to {\em bias} its association to a pico. Before we discuss these mechanisms in
more details, we provide a very brief introduction to format of downlink transmissions
in LTE.

{\em Downlink transmission format:} In LTE, transmissions are scheduled once every
subframe of duration 1~ms; 10 such subframes consist of a frame of length 10~ms.
Each subframe is further divided into 2~slots of duration
0.5~ms each. Each slot consists of 7~OFDMA symbols. While we do not need any further
details for our discussion, the interested reader can refer to~\cite{ltebook} for extensive
details of LTE downlink transmission format. We now describe two important features of
eICIC.

\begin{figure}[t]
\begin{center}
\includegraphics[height=1.7in,width=3.0in]{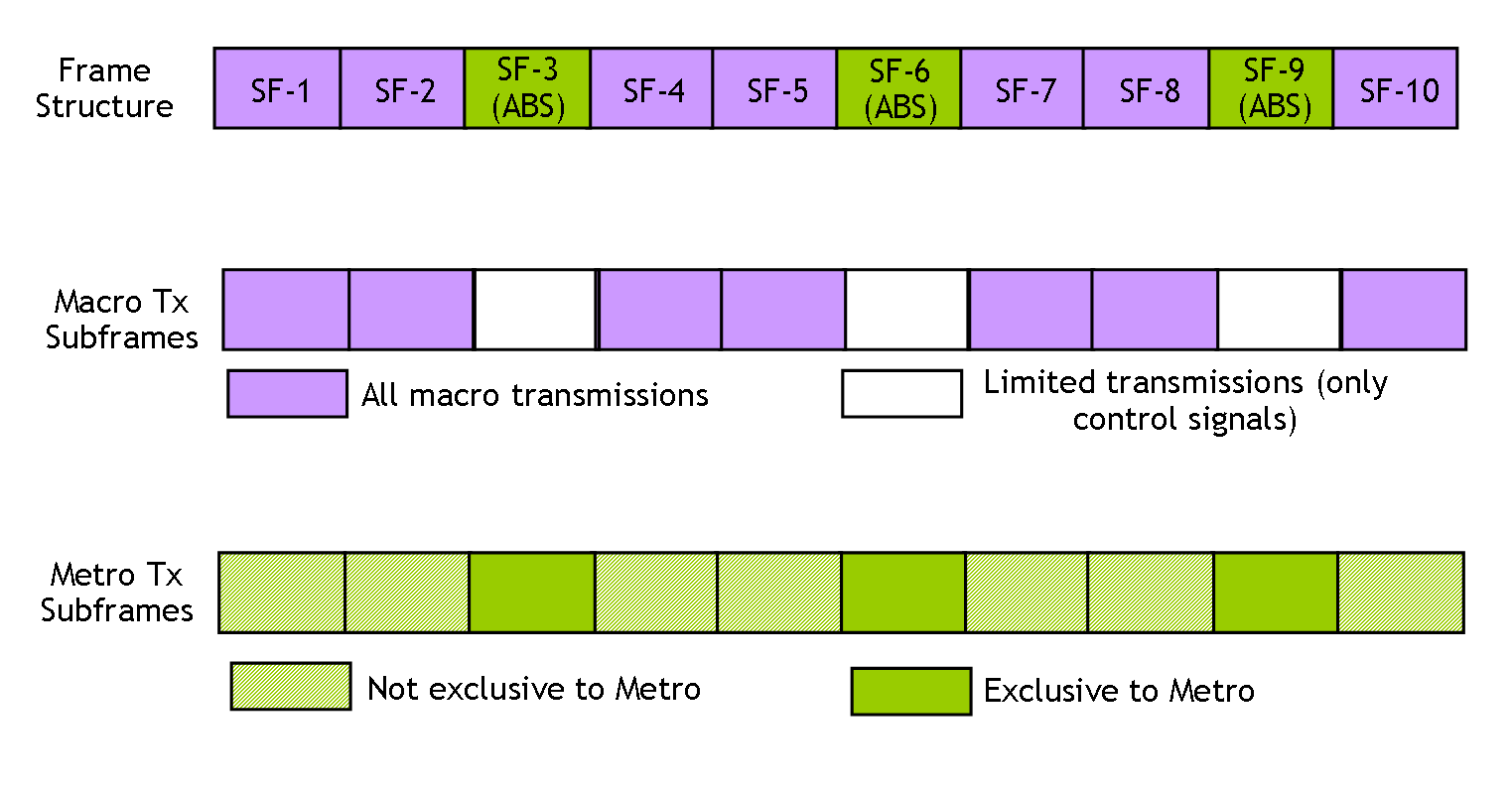}
\caption{\label{fig:abseg}
An illustration of how an LTE frame can consist of ABS subframes. A pico can
transmit over ABS subframes with very little interference from macro, and it can
also transmit over any other non-ABS subframe when it receives high interference from
macro.}
\end{center}
\vspace{-0.20in}
\end{figure}
 
{\bf Almost Blank Subframes (ABS):} In order to assist pico downlink transmissions,
the macro eNodeBs can mute all downlink transmissions to its UEs in certain
subframes termed {\em almost blank subframes} (ABS). These subframes are called
``almost blank" because a macro can still transmit some broadcast signals over
these subframes. Since these broadcast signals only occupy a small
fraction of the OFDMA sub-carriers, the overall interference a macro causes to a
pico is much less during these ABS periods. Thus, the pico can transmit to its
UEs at a much higher data rate during ABS periods. Note that, a pico is also
allowed to transmit to its UEs during non-ABS periods. This could provide good
enough performance to UEs very close to the pico.  An example of ABS schedule is
shown in Figure~\ref{fig:abseg}.

{\bf Flexible User Association and Cell selection bias (CSB):} Typically in cellular
networks, when a UE device (UE) has to select a suitable cell for association, it
chooses the one with maximum received signal strength.  However, if the same strategy
is extended to HetNet deployments with both macro and pico cells, this could lead to
underutilization of the pico eNodeB's. This is because, picos transmit at very low
power and thus, unless a UE is very close to the pico, signal strength from the
macro is likely to be larger for the UE. To overcome this, LTE standards have
proposed a concept called cell selection bias which works in the following manner.
Suppose the cell selection bias of cell-$i$ is $\alpha_i$.  Denote by $P_i$ as the
reference-signal received power (in dBm) from cell-$i$ as measured by a UE.  Here a
cell could be a pico-cell or a macro-cell. Then, the UE associates with cell-$k$
such that \[ k = \arg\max_i (P_i +\alpha_i) .\] Thus, by assigning larger (smaller)
bias to picos compared to macros, once can ensure that the picos are not
underutilized (over-utilized). The bias values are broadcast by the cells to assist
UEs make the right association decision.

\changeclr

\vspace{-0.15in}
\subsection{Connection with Other Interference Mitigation Techniques}
\label{sec:otherintmit}

It is instructive to discuss the connection of eICIC with two other interference mitigation
techniques found in LTE specifications, namely, the frequency domain Inter-cell Interference
Coordination (ICIC) and Coordinated Multi-Point (CoMP).

{\bf Connection with ICIC:} ICIC pre-existed eICIC in terms of both standards
development and initial deployments.  In ICIC based schemes, cell resources are
divided into frequency bands and transmit power profiles to reduce inter-cell
interference~\cite{twcRahmanY10,ShirakabeMM11,sim-icic}. The obvious difference is
that ICIC works in the power spectral density (PSD) domain as opposed to the time
domain of eICIC. ICIC has certain limitations compared to eICIC. In ICIC based
schemes, the different power profiles can either span different sub-bands (thus
leading to a non-uniform frequency re-use factor) or it can span multiple carriers.
The disadvantage of sub-band ICIC is that standards did not include the capability of
varying the {\em reference signal} power according to the sub-band, giving rise to data demodulation
issues for QAM-modulated symbols. Another challenge in deploying sub-band based ICIC
is to find sufficient number of PSD patterns in a dense deployment with macros and
picos.  The carrier based ICIC, on the other hand, introduces significant number of
inter-carrier handovers. The requirement to define multiple carriers also increases
the cost of small cells.  Due to these reasons,  eICIC has attracted much attention
of operators and standard bodies especially for mitigating interference to picos in
HetNets. Having said that, {\em ICIC based macros and eICIC are not either-or
propositions and eICIC benefits can come on the top of existing ICIC based macro
deployments.  Our model and framework easily incorporates existing ICIC based macros;
see Remark~\ref{rem:icic} in Section~\ref{sec:prob} for details.}

{\bf Connection with CoMP:} Coordinated Multipoint Access (CoMP) is also another
edge-rate improving technology for downlink and uplink both. In downlink, CoMP can be
considered an extension of multiuser MIMO (MU-MIMO), that  achieves interference
mitigation by transmitting simultaneously from multiple cells with properly chosen
antenna weights so as to achieve some optimal physical-layer oriented metric like
zero-forcing or joint Minimum Mean Square Error (MMSE).  In spite of theoretical
performance benefits, CoMP faces severe practical challenges due to high backhaul
requirements between base stations, extremely tight time-phase-frequency
synchronization across the set of collaborating cells, and requirement of pico cells
to be equipped with antenna arrays leading to cost increase of supposedly low-cost
picos. Evidently, it will take significant investment and time for CoMP to be widely
deployed in LTE networks that are still in their early days. eICIC on the other hand
is much simpler to deploy and is likely to be adopted much sooner. To this end, our
work optimizes eICIC which does not account for CoMP.  Nevertheless, jointly
optimizing a modified version of eICIC to account for CoMP is a technically
challenging problem and is left as a future work.  See Remark~\ref{rem:comp} in
Section~\ref{sec:prob} for further details.

\normalclr

\vspace{-0.15in}
\subsection{Related Work}
\label{sec:related}

The eICIC proposal is relatively new for LTE Heterogeneous networks.
In~\cite{eicic-intro}, the authors present a very good introduction to the concept of
eICIC in LTE HetNets. In~\cite{survey-eicic}, the authors provide an excellent survey
on eICIC and the motivation behind eICIC proposal in LTE standards. In a
recent work~\cite{sim-eicic}, the authors present simulation studies to understand
the dependence between network performance and eICIC parameters. However, the authors
primarily consider uniform eICIC parameter in all the cells; clearly, the right
choice of eICIC parameters should vary across cells and account for propagation map,
cell-load etc.  Also, the authors do not present a framework to optimize the eICIC
parameters.

In the previous subsection, we have talked about another interference
mitigation technique called ICIC that has been an area of active research in
the recent past.  In~\cite{twcRahmanY10,ShirakabeMM11,sim-icic}, the impact of
ICIC has been studied for LTE and LTE HetNets. The concept of soft-frequency
reuse has been formally studied in~\cite{ffrStolyar08}, where, the authors
optimize downlink transmit power profiles in different frequency bands. The
work closest to ours in principle is~\cite{MadanBSBKJ10} that considers the
problem of UE association and ICIC in a joint manner.

\normalclr

\vspace{-0.15in}
\section{System Model}
\label{sec:model}

\subsection{Terminologies}

In addition to the notion of subframes described in Section~\ref{sec:eicic}, we will use the 
following terminologies in this paper.

{\em UE (user equipment):} UE refers to the mobile device.

{\em eNodeB (eNB):}  The eNB\footnote{eNB is equivalent to base station in
traditional cellular voice networks but it has more functionalities.} is 
the network element that interfaces with the UE and it performs radio resource
management, admission control, scheduling, QoS enforcement, cell information
broadcast etc. It hosts critical protocol layers like PHY, MAC, and Radio Link
Control (RLC) etc. 

{\em Macro cell:} In LTE heterogeneous networks (HetNet), a macro cell 
has a base station transmitter with high transmission power (typically 20~W\--40~W), sectorized
directional antennas, high elevation, and thus ensuring a cell coverage
radius typically around 0.5~km\--2~km.

{\em Pico cell:} As opposed to a macro cell, a pico transmitter is characterized
my much lower transmission power (typically 2~W\--5~W), omnidirectional antennas, low
antenna height, low cost, and has a cell coverage radius of around 100\--300~m. Pico
cells are underlayed on the macro-cellular network to fill coverage holes and to enhance capacity in
traffic hotspot locations. 

{\em Reference Signal Received Power (RSRP):} Every UE in LTE makes
certain measurements of received signal strength of all nearby cell transmitters. RSRP is the
average received power of all downlink reference signals across the entire bandwidth
as measured by a UE. RSRP is taken as a measure of the received signal strength of a
cell transmitter at a UE.


\vspace{-0.15in}
\changeclr
\subsection{Network Model and Interference graph}

\begin{table}[t] \footnotesize
    \caption{List of Parameters and Key Optimization Variables}
    \label{tab:params}
\vspace{-0.2in}
\begin{center}
\begin{tabular}{| c | c |}
    \hline
   {\bf Notation} & {\bf Description}  \\
\hline
\hline
    $\mU,\ u,\ N$     & Set of UEs, index for a typical \\
    $(u\in\mU)$  & UE, number of UEs, respectively \\
\hline
    $\mM,\ m,\ M$     & Set of macros, index for a typical \\
    $(m\in\mM)$  & macro, number of macros, respectively \\
\hline
    $\mP,\ p, P$     & Set of picos, index for a typical\\
    $(p\in\mP)$   & pico, number of picos, respectively \\
\hline
    $m_u $ & The macro that is {\em best} for UE-$u$ \\
\hline
    $r_u^{macro} $ & Data-rate achievable by UE-$u$ from $m_u$ \\
         & (in bits/sub-frame) \\
\hline
    $p_u $ & The pico that is {\em best} for UE-$u$ \\
\hline
    $r_u^{pico,ABS} $ & Data-rate achievable by UE-$u$ from $p_u$ when all \\
                       & interfering macros are muted \\
         & (in bits/sub-frame) \\
\hline
    $r_u^{pico} $ & Data-rate achievable by UE-$u$ from $p_u$ when \\
                     & all/some interfering macros are transmitting \\
         & (in bits/sub-frame) \\
\hline
    ${\mI}_p $ & Set of macro eNB's that interfere \\
                     & with pico $p$ \\
\hline
    ${\mU}_m $ & Set of UEs for whom macro-$m$ is the \\ & best macro eNB\\
\hline
    ${\mU}_p $ & Set of UEs for whom pico-$p$ is the \\ & best pico eNB\\
\hline
\changeclr
$A_p$ & Variable for ABS subframes used/received by pico-$p$  \\
\hline
$N_m$ & Variable non-ABS subframes used by macrio-$m$ \\
\hline
$x_u$ & Variable denoting UE-$u$'s air-time from macro\\
\hline
$y_u^A,\ y_u^{nA}$ & Variables denoting UE-$u$'s air-time from pico \\
                   & over ABS and non-ABS subframes respectively \\
\hline
$R_u$ & Variable denote UE-$u$'s average throughput \\
\hline
{\bf z,\ p} & Vector of all primal and dual variables, respectively \\
\hline
\end{tabular}
\vspace{-0.35in}
\end{center}
\end{table}

Since the eICIC proposal by LTE standard aims to protect downlink
pico transmissions\footnote{The uplink problem (via power
control~\cite{ltefpc}) in presence of picos is not different from macro only
network because UE capabilities from a transmit power point of view remain
same.} and our goal is to develop solutions for optimal eICIC setting, we only
consider downlink transmissions in this work.  \normalclr

{\em Network Topology:} Our system model consists of a network of
macro and pico (also called pico in this paper) eNBs. $\mM$ denotes the set of macros and $\mP$
denotes the set of picos. We also use $m$ and $p$ to denote a typical macro and a
typical pico respectively. 

\changeclr
{\em Interference Modeling and Macro-pico Interference Graph:}  We now describe our
interference model. For the purpose of
eICIC algorithms, it is important to distinguish macro-pico interference from the rest.

\begin{itemize}

\item {\em Macro-pico interference:} For each pico-$p$, the set of macros that
interfere with it is denoted by ${\mI}_p\subseteq \mathcal{M}.$ The macros in the set
$\mI_p$ need to be silent during any ABS subframes used by pico-$p$. Thus, UEs of
pico-$p$ can be interfered by $m\in \mI_p$ only during non-ABS subframes.

\item {\em Macro-macro and pico-pico interference:} Due to
1:1 frequency re-use in most LTE networks, picos can interfere with each other and
similarly for the macros. A pico UE can be interfered by another pico both during ABS and non-ABS
subframes. 

\end{itemize}

To better understand the distinction between the two kinds of interferences from an eICIC
point of view, consider a pico associated UE-$u$'s interfering cells.
Suppose the total interference power it receives from all other interfering picos and all
interfering macros (those in the set ${\mI}_p$) be  $P^{pico}_{Int}(u)$ and $P^{macro}_{Int}(u)$, respectively.
Denoting by $P_{Rx}(u)$ the received downlink power of UE-$u$, the downlink SINR of
UE $u\in \mU_p$ can be modeled as
\begin{align}
\label{eqn:sinrpico}
& \mathsf{SINR}(u) = \\
& \left\{
\begin{array}{ll}
\frac{P_{Rx}(u)}{P^{pico}_{Int}(u) + N_0}\ , & \text{for ABS sub-frames} \\
\frac{P_{Rx}(u)}{P^{pico}_{Int}(u) + P^{macro}_{Int}(u) +  N_0}\ , & \text{for non-ABS sub-frames} 
\end{array} \right. \nonumber
\end{align}
This is because, during ABS sub-frames all interfering macros of pico-$p$ remain silent and
so the only interference is from the interfering picos of $p$. However, during non-ABS
sub-frames, there is interference from all interfering picos and macros both. 
Instead, if UE-$u$ were a macro-UE, the SINR expression would be 
\begin{equation}
\label{eqn:sinrmacro}
\mathsf{SINR}(u) = \frac{P_{Rx}(u)}{P^{pico}_{Int}(u) + P^{macro}_{Int}(u) +  N_0} \ ,
\end{equation}
where $P^{pico}_{Int}(u)$ and $P^{macro}_{Int}(u)$ denote the interference from interfering
picos and macros respectively.

{\em Thus the interference graph relevant from picos'  point of view is the bipartite graph
formed by  joining edges from any $p$ to macros in the set $\mI_p$.} The graph neighbors
of a pico-$p$ should all remain silent during ABS periods usable by picos.

\begin{remark}
The elements of ${\mI}_p$ can be obtained either through 
cell-adjacency relationship or based on whether the received signal from macro eNB to $p$ is 
above a threshold. 
\end{remark}

\normalclr

{\em User model:} To start with, we will consider a scenario where there is a set of
static UE's denoted by $\mU$, and also, we know for each UE-$u$ {\em (i)} the best
candidate macro in terms of RSRP and average PHY data-rate from the macro
$r_u^{macro}$ , {\em (ii)} the best candidate pico, if any, and average PHY data-rate
in ABS and non-ABS subframes given by $r_u^{pico,ABS}$ and $r_u^{pico}$ respectively.  
Note that, the value of $r_u^{macro}$ can be obtained from the SINR
expression~(\ref{eqn:sinrmacro}) using LTE table lookup for conversion from SINR to
rate, and similarly the values of $r_u^{pico,ABS}$ and $r_u^{pico}$ from the SINR
expression ~(\ref{eqn:sinrpico}).  Alternatively, one can use Shannon capacity
formula (with some offset in SINR) to obtain $r_u^{macro}, r_u^{pico,ABS},
r_u^{pico}$ from the corresponding SINR expressions.  Clearly, the average PHY data
rate that a UE receives from pico is higher in ABS frames due to reduced interference
from nearby macros. In fact, the average PHY data rate from a pico in a non-ABS
subframe is likely to be very small in many instances due to very high interference
from macro.

\changeclr 
Note that picos could be deployed indoor or outdoor
(typically outdoor) but makes no difference to our framework. Only the
pico to UE propagation (and hence data rates) change appropriately. \normalclr
The main parameters and some of the optimization variables (introduced later) are captured 
in Table~\ref{tab:params}.

\normalclr

\vspace{-0.15in}
\section{Problem Statement and Computational Hardness}
\label{sec:prob}

We will first develop an algorithm to find optimal ABS and CSB configuration with
static UEs scenario where we have the precise knowledge about number of UEs in
different cells along with PHY data rates. We will describe in
Section~\ref{sec:son}, how {\em Monte-Carlo based techniques can be used along with
our algorithm for this scenario where UE densities and SINR distributions are known
instead of exact UE locations.} 

The essence of eICIC approach is to compute optimal association (either to macro or
to a pico) rules for UEs, and also compute how macros and picos share radio
resources in time domain. Thus, we will first formulate a problem for the optimal
choice of {\em (i)} UE association, i.e., which UEs associate with the best macro
and which ones associate with the best pico, {\em (ii)} the number of ABS subframes
reserved for interfered picos by each macro eNB. We will denote by $N_{sf}$ as
the total number of subframes over which ABS subframes are reserved (typically,
$N_{sf}=40$). We will also refer to the quantity $N_{sf}$ as ABS-period. 

{\em Optimization variables:} Let $N_m$ be the number of sub-frames for which
macro-$m$ can transmit during each ABS-period (clearly, $N_{sf}-N_m$ ABS subframes
are offered by macro-$m$ in each ABS-period). Let $x_u$ be the time-average
air-time\footnote{This time-average airtime can be achieved through MAC scheduling,
particularly weighted proportional-fair scheduling. We are implicitly assuming two
time-scales here: the ABS selection happens at a slow time-scale and MAC scheduling
happens at a fast time-scale resulting in a time-average airtime for each UE.} in
sub-frames per ABS-period UE-$u$ gets from $m_u$, the candidate best macro for
UE-$u$. Note that $x_u$ need not be an integer. The airtime that UE-$u$ gets from a
pico can be during ABS subframes or regular subframes because pico eNBs can transmit
during ABS subframes and regular subframes. To this end, we define $y_u^A$ and
$y_u^{nA}$ as the time-average airtime in subframes per ABS-period UE-$u$ gets from pico
$p_u$ (best candidate pico of UE-$u$) in ABS subframes and regular subframes
respectively. Also, let $R_u$ be the average throughput UE-$u$ achieves.

{\em Optimization constraints:} These are explained below.
\begin{enumerate}

\item {\em Association constraints:}
The association constraint essentially states
that a UE can associate with either the macro or a pico but not both. Thus,
{\small
\begin{align}
\label{eqn:assoc}
\forall\ u\in \mU\ :&
\ \ x_u(y_u^A+y_u^{nA})=0 \\
& x_u\geq 0,\ y_u^A\geq 0,\ y_u^{nA}\geq 0\ \nonumber .
\end{align}}
so that, either the total airtime $u$ gets from macro is zero or the total airtime
$u$ gets from pico is zero.

\item {\em Throughput constraints:} This implies that the average throughput for a UE-$u$,
$R_u$, cannot be more than what is available based on the air-times from associated
macro/pico.
{\small
\begin{align}
\label{eqn:tpt}
\forall\ u\in \mU\ :&
\ \ R_u \leq r_u^{macro}x_u + r_u^{pico,ABS}y_u^A + r_u^{pico}y_u^{nA} \ .
\end{align}}
Note that, from~(\ref{eqn:assoc}), if UE-$u$ is associated to a macro, then 
$R_u\leq r_u^{macro}x_u$; if UE-$u$ is associated to a pico, then
$R_u\leq r_u^{pico,ABS}y_u^A\ +\ r_u^{pico}y_u^{nA}$.

\item {\em Interference constraints:} The interference constraint states that the ABS
subframes used by a pico $p$ are offered by {\em all} macros in the set ${\mI}_p$ that interfere 
with the pico. In other words,
{\small
\begin{align}
\label{eqn:intcon}
\forall\ (p,m\in {\mI}_p):\ \ 
A_p + N_m & \leq N_{sf}\  .
\end{align}}

\item {\em Total airtime constraints:} This ensures that the total time-average airtime 
allocated to UE's from a macro or a pico is less than the total usable subframes. This can be
described using the following inequalities.
{\small
\begin{align}
\label{eqn:xu}
\forall\ m\in \mM\ :\ \ & \sum_{u \in {\mU}_m} x_u  \leq N_m \\ 
\label{eqn:yuA}
\forall\ p\in \mP\ :\ \ & \sum_{u \in {\mU}_p} y_u^A \leq A_p \\  
\label{eqn:yunA}
\forall\ p\in \mP\ :\ \ &\sum_{u\in {\mU}_p}(y_u^A + y_u^{nA}) \leq N_{sf} 
\end{align}}
\end{enumerate}

{\em Optimization objective and problem statement:} The optimization objective we
choose is the tried and tested weighted proportional-fair objective which maximizes
$\sum_u w_u\ln R_u$, where $w_u$ represents a weight associated with UE $u$. This
choice of objective has three benefits. Firstly, it is well known that
proportional-fair objective strikes a very good balance between system throughput and
UE-throughput fairness~\cite{tse}. This bodes well with the goal of improving
cell-edge throughput using eICIC. Secondly, such a choice of objective gels very well
with the underlying LTE MAC where the most prevalent approach is to maximize a
proportional-fair metric. Finally, the weights $w_u$ provides a means for
service-differentiation~\cite{weightedpf01} which is a key element in LTE. The weights may be induced
from policy-determining functions within or outside of radio access network. Though
we use  proportional-fair metric in this paper, our algorithms can be easily
modified to work for other utility functions.

The problem can be stated as follows:

\begin{center}
\fbox{\parbox[c]{0.95\linewidth}{
\begin{center}
{\bf OPT-ABS}\\
\end{center}

{\em Given:} A set of UEs $\mU$, a set of picos $\mP$, and a set of macros $\mM$.
For each UE $u\in \mU$ we are given the following: best candidate parent macro $m_u$
along with PHY data rate to $m_u$ denoted by $r_u^{macro}$,
the best candidate parent pico $p_u$ along with
ABS and non-ABS PHY data rate to candidate parent pico denoted by $r_u^{pico,ABS}$
and $r_u^{pico}$ respectively. We are also given the macro-pico interference graph in
which the interfering macros of pico-$p$ is denoted by $\mI_p$. Finally, $N_{sf}$ is the total
number of subframes.\\

{\em To compute:} We wish to compute the number of ABS subframes $A_p$ each pico-$p$ can
use, the number of non-ABS subframes $N_m$ left for macro-$m$'s usage,
a {\em binary} decision on whether each UE-$u$ associates with its candidate parent pico or
candidate parent macro, throughput $R_u$ of each UE-$u$, so that the following optimization 
problem is solved:
{\begin{align*}
&\ \ \ \ \text{maximize}_{\{x_u,y_u^A,y_u^{nA},A_p,N_m,R_u\}} \sum_u w_u\ln R_u \\
&\ \ \ \ \text{subject to,} \ \ \ \ \  
(\ref{eqn:assoc}),\ 
(\ref{eqn:tpt}),\
(\ref{eqn:intcon}),\
(\ref{eqn:xu}),\ 
(\ref{eqn:yuA}),\ 
(\ref{eqn:yunA})\ \\
\label{eqn:intcon1}
&\ \ \ \ \ \ \ \ \ \ \ \ \ \ \ \ \ \ \ \
 \forall\ (p,m\in {\mI}_p):\ \ A_p,\ N_m  \in \mathbb{Z}^+\ , 
\end{align*}}
where $\mathbb{Z}^+$ denotes the space of non-negative integers. 
}}\\
\end{center}

We also call the optimization objective as {\em system
utility} and denote it, as a function
of the all UE's throughput-vector ${\mathbf R}$, by
\[\text{Util}({\mathbf R})\ =\ \sum_u w_u \ln R_u\ .\]

\begin{remark} 
\label{rem:icic}({\sc Accounting for ICIC})
Though eICIC or time-domain resource sharing is the preferred mode of resource
sharing between macro and pico cells in LTE for reasons mentioned in
Section~\ref{sec:otherintmit}, eICIC could co-exist with ICIC in macro-cells.
Our framework can easily account for this with only changes in the input to the
problem. In ICIC, the OFDMA sub-carriers of a macro are partitioned into two parts: low-power
sub-carriers of total bandwidth $B_l$ and high-power sub-carriers of total bandwidth
$B_h$. Thus for a UE $u$ that receives signal from macro-$m$, the spectral
efficiency over low-power subcarriers (say, $\eta_u^{low}$) is different from spectral
efficiency over high-power
subcarriers (say, $\eta_u^{high})$. This causes the downlink rate from the candidate
parent macro,
$r_u^{macro}$, to be expressed as $r_u^{macro}=\eta_u^{high} B_h + \eta_u^{low} B_l$.
Furthermore, the reduced interference from ICIC-using macros changes
the macro-pico interference graph structure and pico to UE non-ABS rates.
Importantly, {\bf only the input to OPT-ABS has to be modified to account for ICIC in
frequency domain.}
\end{remark}

\begin{remark}
\label{rem:comp}({\sc CoMP})
Our framework assumes no CoMP based deployments. This would be the case in the most
LTE deployments in the foreseeable future due to practical challenges of CoMP
outlined in Section~\ref{sec:otherintmit}. Technically speaking, to optimize a
modified version of eICIC that accounts for CoMP, the association constraint given 
by~(\ref{eqn:assoc}) would not be required,
the throughput constraint~(\ref{eqn:tpt}) must account for collaborating cells in
CoMP, and the total airtime constraints must be modified to reflect CoMP. We leave
this as a future work.
\end{remark}

\normalclr

\vspace{-0.15in}
\subsection{Computational hardness}

It can be shown that the ABS-optimization problem is NP-hard even with a single macro
but multiple picos. We state the result as follows.

\begin{proposition}
\label{THM:NPHARD}
Even with a single pico and a single interfering macro, the OPT-ABS problem is NP-hard unless $P=NP$.
\end{proposition}

\begin{proof}
Follows by reducing the SUBSET-SUM problem~\cite{gj} to an instance of
OPT-ABS problem. See\cite{eiciclonger}. 
\end{proof}

In light of the above result, we can only hope to have algorithm that is provably
a good approximation to the optimal. In the following, we will develop an algorithm
with a constant-factor worst case guarantee; we show extensive simulation results to
demonstrate that our algorithm is within 90\% of the optimal is many practical
scenarios of interest.


\vspace{-0.15in}
\section{Algorithm Overview}
\label{sec:algo1}

Our approach to the problem is to solve it in two steps.
\begin{enumerate}

\item {\em Solving the relaxed NLP:} In the first step, we solve
the non-linear program (NLP) obtained by ignoring integrality constraints on $A_p$ and
$N_m$ and also the constraint that a UE can receive data either from pico or macro
but not both. Specifically, in this step, we maximize $\text{Util}({\mathbf R})$ subject to the
constraints (\ref{eqn:tpt})\--(\ref{eqn:yunA}); and we
allow $A_p$ and $N_m$ to take non-integer values.  Note that ignoring the
constraint~(\ref{eqn:assoc}) means that UEs can receive radio resources from macro
and pico both.
For notational convenience, we also denote the vector of constraints 
(\ref{eqn:tpt})\--(\ref{eqn:yunA}), in a compact form as 
${{\mathbf g}}_R(.)\leq 0\ .$ Thus, the RELAXED-ABS problem can be denoted as,
\begin{align*}
& RELAXED-ABS: \\
&\ \ \ \ \text{maximize}_{\{x_u,y_u^A,y_u^{nA},A_p,N_m\}} \sum_u w_u\ln R_u \\
&\ \ \ \ \text{subject to,} \ \ \ \ \  {{\mathbf g}}_R(.)\leq 0 \\
&\ \ \ \ \forall\ (p,m\in {\mI}_p):\ \ \ \ A_p,\ N_m  \in \mathbb{R}^+\ , 
\end{align*}
where $\mathbb{R}^+$ denotes the space of non-negative real numbers and
${{\mathbf g}}_R(.)$ denotes the vector of constraints
(\ref{eqn:tpt})\--(\ref{eqn:yunA}).
 
\item {\em Integer rounding:} In the second step, we appropriately round the output
of the non-linear optimization to yield a solution to the original problem that is
feasible.

\end{enumerate}

\vspace{-0.15in}
\section{Algorithm for Relaxed Non Linear Program} 
\label{sec:algo2}

Towards solving ABS-RELAXED, we use a dual based approach~\cite{asusubgrad,bertNLP}
which has been successfully applied to many networking problems for its simplicity of
implementation~\cite{chelow06,srikantmoi}.  In the following, we show that, a dual
based approach to our problem lends to a decomposition that greatly reduces the
algorithmic complexity and also makes the approach amenable to distributed
implementation while retaining the core essence.

In a dual based approach, it is crucial to define a suitable notion of feasible
sub-space so that the solution in each iteration is forced to lie within that
sub-space. The choice of the sub-space also has implication on the convergence speed
of the algorithm. To this end, we define the sub-space $\bPi$ as follows:

\begin{align}
\bPi=&\{{\mathbf {x,y,A,N}}: A_p\leq N_{sf}, N_m\leq N_{sf},
\sum_{u\in {\mU}_m}x_u\leq N_{sf}, \nonumber \\
& \sum_{u\in {\mU}_p}y_u^A\leq N_{sf},
\sum_{u\in {\mU}_p}y_u^{nA}\leq N_{sf},\ \ \forall\ m,p
\}
\end{align}

We have used bold-face notations to denote vectors of variables.  
Clearly, any solution that
satisfies the constraints described in the previous section lies within $\bPi$.  In
the following discussion, even without explicit mention, it is understood that
optimization variables always lie in $\bPi$.

We now describe the non-linear program (NLP) obtained by treating $A_p$ and $N_m$ as
real numbers.  The Lagrangian of the relaxed NLP can be expressed as follows:
{\small
\begin{align}
\label{eqn:lag}
&\mathcal{L}({\mathbf{x,y,A,N, {\boldsymbol{\lambda,\mu,\beta,\alpha}}}})
= \sum_u w_u\ln R_u  \\
&-\sum_u \lambda_u (R_u - r_u^{macro}x_u - r_u^{pico,ABS}y_u^A - r_u^{pico}y_u^{nA})
\nonumber\\
&- \sum_{p,m\in {\mI}_p}\mu_{p,m}(A_p + N_m - N_{sf})
\nonumber\\
& -\sum_m \beta_m (\sum_{u \in {\mU}_m} x_u  - N_m)  - \sum_p \beta_p (\sum_{u \in {\mU}_p} y_u^A - A_p)  
\nonumber\\
&- \sum_p \alpha_p(\sum_{u\in {\mU}_p}(y_u^A + y_u^{nA}) - N_{sf}) 
\nonumber 
\end{align}}

{\em Additional notations:} We use bold-face notations to express vectors. For example
$\boldsymbol{\lambda}$ denotes the vector of values $\lambda_u$. The variables
$\lambda,\mu,\beta,\alpha$'s are dual variables and so called Lagrange-multipliers
which also have a price interpretation.  In the rest of the paper, $\bp$ denotes the
vector of all dual variables, i.e., $\bp=(\lambda,\mu,\beta,\alpha)$\footnote{The
dual variables are often denoted by $\bp$ because they have the interpretation of
prices.}. Similarly, the variables  $\mathbf {x,y,A,N}$ are referred to as primal
variables and we use $\bz$ to denote the vector of all primal variables, i.e.,
$\bz=(\mathbf {x,y,A,N})$. Thus, we denote the Lagrangian by $\mL(\bz,\bp)$ and
express it as
\[ \mL(\bz,\bp) = \text{Util}({\mathbf R})\ -\ {\mathbf p}' {{\mathbf g}}_R({\mathbf z})\ .\]


The dual problem of RELAXED-ABS can be expressed as
\begin{equation}
\label{eqn:dual}
\min_{\bp\geq \mathbf{0}}
\mathcal{D}(\bp)\ ,
\end{equation}
where,
\begin{equation}
\label{eqn:dualcost}
\mathcal{D}(\bp)
=\max_{\bz \in \bPi} 
\mathcal{L}(\bz,\ \bp)\ .
\end{equation}
Since the RELAXED-ABS is a maximization problem with concave objective and convex feasible
region, it follows that there is no duality gap~\cite{bertNLP}, and thus
\[\text{RELAXED-ABS Optimal}
=\min_{\bp}
\mathcal{D}(\bp)\ . \]

{\bf Iterative steps:} 
First the primal variables are
initialized to any value with $\bPi$ and the dual variables are initialized to zero,
and then, the following steps are iterated (we show the update for iteration-$t$):

\begin{enumerate}

\item {\em Greedy primal update:} The primal variables ${\mathbf z}_t$ in
iteration-$(t+1)$ are set as
\begin{equation}
\label{eqn:pu}
{\mathbf z}_{t+1}\ =\ \arg\max_{{\mathbf z}\in \bPi} \mathcal{L}(\bz,\ \bp_{t})\ .
\end{equation}

\item {\em Subgradient descent based dual update:} The dual variables are updated in a
gradient descent like manner as
\begin{equation}
\label{eqn:du}
\bp_{t+1} =[\bp_{t} +\gamma {\mathbf g}_R({\mathbf z}_{t})]^+\ , 
\end{equation}
where $\bp_t$ is the dual variable at iteration-$t$, $\gamma$ is the step-size,
and $[.]^+$ denotes component-wise projection into the space of non-negative real numbers. 


\end{enumerate}

The above steps are continued for sufficiently large number of iterations $T$
and the optimal solution to RELAXED-ABS is produced as
\[ {\hat{{\mathbf z}}}_T=\smfrac{1}{T}\sum_{t=1}^{T}{\mathbf z}_t\ .\]

The computation of {\em greedy primal update} is not immediate at a first glance, however, the
{\em subgradient descent based dual update} step is straightforward. Thus,
for the above algorithm to work, there are two important questions that need to be
answered: {\em (i)} how can the {\em greedy primal update} step be performed
efficiently? {\em (ii)} how should the step size $\gamma$ and the number of
iterations $T$ be chosen? In the following, we answer these questions.

\vspace{-0.15in}
\subsection{Greedy Primal Update: Decomposition Based Approach}

We now argue that the problem of computing 
\[\arg\max_{{\mathbf z}\in \bPi} \mathcal{L}(\bz,\ \bp_{t-1})\]
can be decomposed into UE problem, macro problem and pico problem each of which is
fairly straightforward. Towards this end, we rewrite $\mathcal{L}(\bz,\ \bp)$ as
follows.

{\small
\begin{align*}
\mathcal{L}(\bz,\bp) & = \sum_u F_u(\bp,R_u)
+\sum_m G_m(\bp,\{x_u\}_{u\in {\mU}_m},N_m)\\ 
&\ \ +\sum_p H_p(\bp,\{y_u\}_{u\in {\mU}_p},A_p) -N_{sf}\ ,
\end{align*}
}
where,
{\small
\begin{align*}
& F_u(\bp,R_u)
=w_u\ln R_u-\lambda_uR_u \\
& G_m(\bp,\{x_u\}_{u\in {\mU}_m},N_m)=N_m(\beta_m-\sum_{p:m\in {\mI}_p}\mu_{p,m})\\
&+\ \sum_{u\in {\mU}_m}x_u(\lambda_ur_u^{macro}-\beta_{m})\\
& H_p(\bp,\{y_u\}_{u\in {\mU}_p},A_p)= A_p(\beta_p-\sum_{m:m\in {\mI}_p}\mu_{p,m})\\
&+\ \sum_{u\in {\mU}_p}y_u^A(\lambda_ur_u^{pico,ABS}-\beta_{p}-\alpha_{p})
+\ \sum_{u\in {\mU}_p}y_u^{nA}(\lambda_ur_u^{pico}-\alpha_{p})
\end{align*}
}

It follows that,
{\small
\begin{align*}
&\max_{{\mathbf z}\in \bPi} \mathcal{L}(\bz,\bp) \\
& = \sum_u \max_{{R_u}} F_u(\bp,R_u)
 +\sum_m \max G_m(\bp,\{x_u\}_{u\in {\mU}_m},N_m) \\
&\ \ \ +\sum_p \max H_p(\bp,\{y_u\}_{u\in {\mU}_p},A_p) -N_{sf}
\end{align*}
}
where the $\max$ in the above is with respect to appropriate primal variables from $x_u, y_u$'s,
$N_m$'s and $A_p$'s. The above simplification shows that {\em greedy primal update}
step can be broken up into sub-problems corresponding to individual UEs,
individual macros, and individual picos; each of the sub-problems has a
solution that is easy to compute as follows. We thus have the following.

{\bf Greedy primal update:} In iteration-$t$, the greedy primal updates are as
follows:
\begin{itemize}

\item {\em User primal update:} 
In {\em greedy primal update} step of iteration-$(t+1)$, 
for each UE-$u$, we maximize $F_u(\bp_{t},R_u)$ by choosing $R_u(t+1)$ as
\begin{equation}
\label{eqn:upd}
R_u(t+1) = \smfrac{w_u}{\lambda_u(t)}\ .
\end{equation}

\item {\em Macro primal update:}
In {\em greedy primal update} step of iteration-$(t+1)$, for
each macro-$m$, we maximize $G_m(\bp_{t},\{x_u\}_{u\in {\mU}_m},N_m)$ by choosing
$N_m$ as
\begin{equation}
\label{eqn:mpd1}
N_m(t+1)  =
N_{sf}\ind{(\beta_m(t)-\sum_{p:m\in {\mI}_p}\mu_{p,m}(t)>0)}\ .
\end{equation}
To compute all $\{x_u\}_{u\in {\mU}_m}$, each macro-$m$ computes the best UE $u_m^*$ in
iteration-$t$ as 
$$ u_m^* = \arg\max_{u\in {\mU}_m}(\lambda_u(t) r_u^{macro}-\beta_{m}(t)>0) $$
where ties are broken at random.
Macro-$m$ then chooses $x_u(t+1), u\in {\mU}_m$ as 
\begin{equation}
\label{eqn:mpd2}
x_u(t+1)=
\left\{
\begin{array}{ll} 
N_{sf} & \text{for}\ u=u_m^* \\ 
0 & \text{for}\ u\neq u_m^*
\end{array}
\right.
\end{equation}

\item{\em Pico primal update:}
In iteration-$(t+1)$, for each pico-$p$, we maximize 
$H_p(\bp_{t},\{y_u\}_{u\in {\mU}_p},A_p)$ by choosing
\begin{equation}
\label{eqn:ppd1}
A_p(t+1)  =
N_{sf}\ind{(\beta_p(t)-\sum_{m:m\in {\mI}_p}\mu_{p,m}(t)>0)}\ .
\end{equation}
To compute all $\{y_u\}_{u\in \mU_p}$'s, each pico-$p$ computes the current best UE 
$u_{p}^*(ABS)$ and $u_p^*(nABS)$ as follows:
$$ u_{\bp}^*(ABS) = \arg\max_{u\in {\mU}_p}
(\lambda_ur_u^{pico,ABS}-\beta_{p}(t)-\alpha_{p}(t)>0)\ ,$$
$$ u_{\bp}^*(nABS) = \arg\max_{u\in {\mU}_p}
(\lambda_u(t)r_u^{pico}-\alpha_{p}(t)>0)\ .$$
where 
ties are broken at random. Pico-$p$ then chooses $y_u(t+1), u\in {\mU}_p$ as 
\begin{equation}
\label{eqn:ppd2}
y_u^A(t+1)=
\left\{
\begin{array}{ll} 
N_{sf} & \text{for}\ u=u_{\bp}^*(ABS) \\ 
0 & \text{for}\ u\neq u_{\bp}^*(ABS)
\end{array}
\right.
\end{equation}
Similarly, we set $y_u^{nA}(t)$ based on $u_{\bp}^*(nABS)$ as follows.
\begin{equation}
\label{eqn:ppd3}
y_u^{nA}(t+1)=
\left\{
\begin{array}{ll} 
N_{sf} & \text{for}\ u=u_{\bp}^*(nABS) \\ 
0 & \text{for}\ u\neq u_{\bp}^*(nABS)
\end{array}
\right.
\end{equation}

\end{itemize}

\vspace{-0.15in}
\subsection{Overall Algorithm for RELAXED-ABS}

We now summarize the algorithmic steps for solving RELAXED-ABS.
Algorithm~\ref{algo:abs-relax} formally describes our algorithm.

\begin{algorithm}[t]
\caption{\textsc{Optimal RELAXED-ABS}:
Algorithm for Solving RELAXED-ABS}
\label{algo:abs-relax} 
{\small
\begin{algorithmic}[1]

\STATE {\em Initialization:} Initialize all the variables
$\boldsymbol{x, y, A, N, \lambda, \mu, \beta, \alpha}$ to any feasible value.

\FOR {$t=0,2,3,\hdots, T$ iterations}

\STATE {\em Primal update:} Update the primal variables ${\mathbf z}(t)$ by
using UE's update given by~(\ref{eqn:upd}), macro updates given
by~(\ref{eqn:mpd1}),~(\ref{eqn:mpd2}), and pico updates given
by~(\ref{eqn:ppd1}),~(\ref{eqn:ppd2},~(\ref{eqn:ppd3}).

\STATE {\em Dual Update:} For each UE-$u$ $\lambda_u(t)$ is updated as
$$\lambda_u(t)\leftarrow [\lambda_u(t-1)+$$
$$\ \ \ \ \gamma(R_u(t) - r_u^{macro}x_u(t) - r_u^{pico,ABS}y_u^A(t) -
r_u^{pico}y_u^{nA}(t))]^+\ .$$

For each macro-$m$, we update its dual price $\beta_m$ as follows:
$$\beta_m(t)\leftarrow [\beta_m(t-1) +\gamma(\sum_{u \in {\mU}_m} x_u(t)  -
N_m(t))]^+$$ 

For each pico $p$, we update all dual variables $\beta_p,\alpha_p$ and $\mu_{p,m}$ for all $m\in {\mI}_p$, 
as follows:
\begin{align*}
\mu_{p,m}(t) & \leftarrow [\mu_{p,m}(t-1)+\gamma(A_p(t) + N_m(t) - 
N_{sf})]^+ \\
\beta_p(t) & \leftarrow [\beta_p(t-1)+\gamma (\sum_{u \in {\mU}_p} y_u^A(t) -
A_p(t))]^+
\\ 
\alpha_p(t) & \leftarrow [\alpha_p(t-1) +\gamma(\sum_{u\in {\mU}_p}(y_u^A(t) + y_u^{nA}(t))
- N_{sf})]^+
\end{align*}

\ENDFOR

\STATE The optimal values of the NLP are obtained by averaging over all iterations:
\begin{equation*}
\hat{{\mathbf z}}_T=\smfrac{1}{T}\sum_{t=1}^{T}{\mathbf z}_t,\ 
\end{equation*}

\end{algorithmic}
}
\end{algorithm}

We next derive the step-size and sufficient number of iterations in terms of the
problem parameters.


\vspace{-0.15in}
\subsection{Step-size and Iteration Rule using Convergence Analysis}

Towards the goal of estimating the step-size and number of iterations, we adapt
convergence analysis for a generic dual based algorithm is provided
in~\cite{asusubgrad}. We show that the structure of ABS-RELAXED lends to a simple
characterization of the step-size and number of iterations in terms of problem
parameters.


In this section, we denote by $r_{max}$ and $r_{min}$ as the maximum and minimum data
rate of any UE respectively. We also denote by $W_m$ as the total weight of all
candidate UEs of macro-$m$ and similarly for $W_p$. Also ${\mathbf W}$ denotes the
vector of $W_m$'s and $W_p$'s. Also $U_m$, $U_p$, $U_{max}$ denote, respectively, the
number of candidate UEs in macro-$m$, number of candidate UEs in pico-$p$, and
maximum number of UEs in any macro or pico.

\begin{proposition}
\label{THM:DUALCONV}

Let $\bz_t,\hat{\bz}_t,{\bz}^*$ ($\bp_t,\hat{\bp}_t,{\bp}^*$) denote the vector of primal
(dual) variables at time $t$, averaged over all iterations from $0\--t$, and at
optimality, respectively. Under mild technical assumptions~\cite{eiciclonger}, 
we have the following:
{\small
\begin{align*}
{(i)}&\ \  \mD(\hat{\bp}_T)-\mD({\bp}^*)\leq \smfrac{B^2}{2\gamma T} 
+ \smfrac{\gamma Q^2}{2}\\
{(ii)}&\ \ 
\text{Util}({\mathbf{R}}^*)-\text{Util}(\hat{\mathbf{R}}_T)\leq \smfrac{\gamma Q^2}{2} 
\end{align*}}
where
{\small
\begin{align*}
Q^2 & = N_{sf}^2(N r_{max}^2 + M + P + 2I ) \\
B^2 & = 
\smfrac{\normsq{{\mathbf W}}}{N_{sf}^2}(1+2I_{max}+\smfrac{U_{max}}{r_{min}})\ .
\end{align*}}
\end{proposition}
\begin{proof}
See \cite{eiciclonger}. 
\end{proof}

\changeclr

\begin{remark}({\sc On the proof of Proposition~\ref{THM:DUALCONV}})
The main contribution of the proof of Proposition~\ref{THM:DUALCONV} is to show that the
norm of optimal dual variable $\normsq{{\mathbf p}^*}$ can be upper bounded by
network parameters. This upper bound, along with adaptation of convergence analysis
in~\cite{asusubgrad}, readily characterizes the step-size and number of
iterations for RELAXED-ABS; this is unlike arbitrary
convex programs where the convergence results are in terms of a generic {\em
slater vector}~\cite{asusubgrad}.
\end{remark}

\normalclr



\begin{remark}
({\sc Step-size and Number of Iterations}.)
\label{rem:ss}
The step-size and number of iterations can be set based on the two following
principles:
\begin{enumerate}
\item 
Suppose we want the per-UE objective to
deviate from the optimal by no more than $\epsilon$. Then,
Proposition~\ref{THM:DUALCONV} can be used to set $\gamma$ and $T$ as follows:
\begin{align}
\smfrac{\gamma Q^2}{2}\leq \smfrac{N\epsilon}{2}\ 
\ \text{and}\ \
\smfrac{B^2}{2\gamma T}\leq \smfrac{N\epsilon}{2}\ .
\label{eqn:ssrule1}
\end{align}
The above imply,
\begin{equation}
\gamma = \smfrac{N\epsilon}{Q^2}
\ \text{and}\ \
T=(\smfrac{QB}{N\epsilon})^2\ .
\label{eqn:ssrule2}
\end{equation}
Since the maximum interferers $I_{max}$ is typically a small number, 
a moments reflection shows that $\gamma = O(\epsilon/N_{sf}^2 r_{max}^2)$
and $T=O(w_{max}^2 U_{max}/\epsilon^2 r_{min})$ where $w_{max}$ is the maximum vaue
of $w_u$. In other words, the number of
iterations simply depends on the maximum UEs in any cell and not on the overall
number of UEs.

\item The number of iterations required can be significantly reduced using the
following observation. Suppose the macro-pico interference graph can be decomposed into
several disjoint components. In that case, we can run the RELAXED-ABS algorithm for
each component independently (possibly parallaly). For each interference graph
component, we can use the step-size and iteration rule prescribed in the previous
paragraph.
\end{enumerate}
\end{remark}


\vspace{-0.15in}
\section{Integer Rounding of RELAXED-ABS}
\label{sec:algo3}

In this section, we show how solution to RELAXED-ABS can be converted to a feasible
solution for the original problem OPT-ABS. There are two challenges in performing
this step.  Firstly, in OPT-ABS, each UE can receive resources either from a macro
or a pico but not both unlike RELAXED-ABS.  Secondly, as with all dual based
sub-gradient algorithms, after $T$ iterations of running RELAXED-ABS, the solution
may violate feasibility, {\em albeit} by a small margin~\cite{asusubgrad}. Thus, we
need to associate each UE with a macro or a pico and round the values of $N_m$'s
and $A_p$'s so that the overall solution is feasible and has provable performance
guarantee.

To this end, we first introduce the following rounding function:
\begin{equation}
\text{Rnd}_{N_{sf}}(x)=\left\{
\begin{array}{ll}
\lfloor x \rfloor &\ ,\ x\geq \smfrac{N_{sf}}{2}\\
\lceil x \rceil &\ ,\ x < \smfrac{N_{sf}}{2}
\end{array}
\right.
\label{eqn:rf}
\end{equation}
The rounding algorithm is formally described in Algorithm~\ref{algo:abs-round}.

\begin{algorithm}[tbph]
\caption{\textsc{Round RELAXED-ABS}:
Algorithm for Integer Rounding of Output of Algorithm~\ref{algo:abs-relax}
\label{algo:abs-round}}
{\small
\begin{algorithmic}[1]

\STATE {\em User Association}: For all $u\in\mU$, perform the following steps:
\begin{enumerate}
\item Compute the throughput $u$ gets from macro
and pico in RELAXED-ABS solution as follows:
\begin{align*}
R_u^{macro} & =r_u^{macro}\hat{x}_u \\
R_u^{pico} & = r_u^{pico,ABS}{\hat{y}}_u^A + r_u^{pico}{\hat{y}}_u^{nA}\ ,
\end{align*}
where, ${\hat{x}}_u, {\hat{y}}_u^A, {\hat{y}}_u^{nA}$ are the out
of Algorithm~\ref{algo:abs-relax}.

\item If {$R_u^{macro}>R_u^{pico}$}, UE-$u$ associates with the macro, else
it with pico.
\end{enumerate}
Define and compute ${\mU}_m^*$, the set of UEs
associated with macro-$m$ after the {\em UE association} step. Similarly define
and compute ${\mU}_p^*$ for every pico-$p$. 

\STATE {\em ABS Rounding:}  Compute integral $N_m^*$'s and $A_p^*$'s as follows:
\begin{align*}
N_m^*&
=\text{Rnd}_{N_{sf}}({\hat{N}}_m)
\ \ \ \ \forall m \in\mM \\
A_p^*&
=\text{Rnd}_{N_{sf}}({\hat{A}}_p)
\ \ \ \ \forall p \in\mP
\end{align*}
where ${\hat{N}}_m$ and ${\hat{A}}_p$ denote the output of
Algorithm~\ref{algo:abs-relax}.

\STATE{\em Throughput computation:}
For each macro-$m$, for all $u\in {\mU}_m^*$, the final value of $x_u^*, R_u^*$ are
\begin{equation}
x_u^* = \smfrac{{\hat{x}}_u N_m^*}{X_m}\ \text{and}\ R_u^*=r_u^{macro}x_u^*\ .
\label{eqn:xround}
\end{equation}
where $X_m=\sum_{u\in\mU_m^*} {\hat{x}}_u\ .$ 

For every pico, compute the ABS-utilization $Y_p^A$ and
non-ABS utilization $Y_p^{nA}$ as
$$Y_p^A=\sum_{u\in \mU_p^*}{\hat{y}}_u^A\ ,\ 
Y_p^{nA}=\sum_{u\in \mU_p^*}{\hat{y}}_u^{nA}\ .$$
Next, for each pico-$p$, for all $u\in {\mU}_m^*$, the final values of
${y_u^A}*,{y_u^{nA}*},R_u^*$
are
\begin{align}
{y_u^A}^* &= \smfrac{{\hat{y}}_u^A A_p^*}{Y_p^A}\ ,\ 
{y_u^{nA}}^* = \smfrac{{\hat{y}}_u^{nA} (N_{sf}-A_p^*)}{Y_p^{nA}}\
\label{eqn:yround}\\
R_u^*\ \  &= r_u^{pico,ABS}{{y}_u^A}^* + r_u^{pico}{{y}_u^{nA}}^*
\label{eqn:rpico}
\end{align}

The system utility is computed as $\text{Util}({\mathbf R}^*)=\sum_u w_u\ln
R_u^*.$

\end{algorithmic}
}
\end{algorithm}

The algorithm has
three high-levels steps. In the {\em UE association} step, each UE who gets
higher throughput from a macro in the solution of ABS-RELAXED is associated with a
macro, and, each UE who gets higher throughout from a pico gets associated with a
pico. In the next step called {\em ABS rounding}, the UE association decisions are
used to obtain the ABS and non-ABS subframes. Indeed, this step produces a feasible
$A_p$'s and $N_m$'s as we show later in our result. Finally, in the throughput
computation step, each UE's available average airtime is scaled to fill-up the
available sub-frames. The throughput of each UE can be computed from this.

{\em Performance guarantees:} The worst case performance guarantee of the output
produced by Algorithm~\ref{algo:abs-round} depends on the number of iterations and
step-size used for running Algorithm~\ref{algo:abs-relax} prior to running
Algorithm~\ref{algo:abs-round}. This is shown by the following result.
Let $R_u^*$ be the throughput computed by Algorithm~\ref{algo:abs-round} and let
$R_u^{opt}$ be the optimal throughput. 

\begin{proposition}
\label{THM:APX}
Algorithm~\ref{algo:abs-round} produces feasible output to the problem OPT-ABS.
Furthermore, for any given $\delta>0$, there exists $T$ large enough (but, polynomial
in the problem parameters and $1/\delta$) and $\gamma$ satisfying~(\ref{eqn:ssrule1})
such that, if we apply Algorithm~\ref{algo:abs-round} to the output of
Algorithm~\ref{algo:abs-relax} with this $T,\gamma$, then
\[\text{Util}(2(1+\delta){\mathbf R}^*)\geq \text{Util}({\mathbf R}^{opt})\]
\end{proposition}
\begin{proof}
See~\cite{eiciclonger}.
\end{proof}

\begin{remark}
Proposition~\ref{THM:APX} shows that, 
for sufficiently large but polynomial number of iterations, the worst
case approximation factor is close to~2. It is important to realize that,
as with all NP-hard problems, this is simply a worst case result.
Our evaluation with several real topologies suggest that the performance of our
algorithm is typically within $90\%$ of the optimal. Also, in practice, we recommend
using the step-size and number of iterations given by~(\ref{eqn:ssrule2}).
\end{remark}


\vspace{-0.15in}
\section{Computing Cell Selection Bias and ABS Patterns}
\label{sec:csbabspattern}

In this section, we describe two important computations that are relevant for
realization of eICIC, namely, Cell Selection Bias (CSB) based UE association and
converting ABS numbers into ABS patterns.

\vspace{-0.15in}
\subsection{Cell Selection Bias for UE Association}
\label{sec:csb}
Our solution so far solves the coupled problem of optimizing ABS sub-frames and UE
association. However, UE association in LTE HetNets need to be standard compliant.
While standards on this are evolving, one proposed methodology by the LTE standard is a 
rule based on {\em cell selection
bias}~\cite{eicic-intro}. Precisely,
if $b_c$ is the cell-selection bias of cell-$c$ (which could be a macro or apico),
then UE-$u$ associates with cell $c_a$ such that,
\begin{equation}
c_a = \max_c [b_c\ +\ RSRP_{u,c}]\ ,
\label{eqn:csbrule}
\end{equation} 
where $RSRP_{u,c}$ is the received RSRP of cell-$c$ at UE-$u$. The choice of $c_a$ is a
design issue. We wish to choose values of $b_c$ so that UE association based on rule
given by~(\ref{eqn:csbrule}) leads to association decisions derived in the previous
section by our optimization algorithm. 

Obtaining cell specific biases that precisely achieves a desired association may not always be
feasible. Thus, we propose to compute biases so that the ``association error" (as
compared to optimal association) is minimized. We do this using the following steps.

\begin{enumerate}

\item Since the relative bias between picos and macros matter, set all biases of
macros to zero.

\item Let $C_{p,m}$ be the set of UEs who have pico-$p$ as
the best candidate pico, and also, macro-$m$ as the best candidate macro, i.e., $C_{p,m}={\mU}_p\cap
{\mU}_m$.
From the UEs in the set $C_{p,m}$, let $W_{p,m}^*$ be the
total weight of UE's associated to pico-$p$ under UE association produced by our algorithm
in the previous section. Also from UEs in $C_{p,m}$, as a function 
of bias $b$,  let $W_{p,m}(b)$ be the total weight of UEs that would associate with
pico-$p$ if the bias of pico-$p$ were set to $b$. In other words, $W_{p,m}=\sum_{u\in
D_{p,m}(b)}w_u$ where
\[ D_{p,m}(b) = \{u\in C_{p,m}: RSRP_{u,p} + b\geq RSRP_{u,m} \}\ , \]
where $RSRP_{u,m}$ and $RSRP_{u,p}$ are received power (in dBm) of reference signal at
UE-$u$ from best candidate macro and best candidate pico respectively.
This step computes $W_{p,m}^*$ and $W_{p,m}(b)$ for every interfering pico-macro
pair $(p,m0$ and every permissible bias value. 

\item For every pico-$p$, cell selection bias $b_p$ is set as
\begin{equation}
\label{eqn:csbmmse}
b_p =\arg\min_b\left[ \sum_{m\in {\mI}_p}{\card{W_{p,m}(b)-W_{p,m}^*}}^2\right] \ .
\end{equation}
Thus, the bias values are chosen as the one that minimizes the mean square error of
the association vector of number of UEs to different picos.

\end{enumerate}

\begin{remark}
({\sc Maximum and minimum bias constraint}.)
In many scenarios, operators that deploy picos may desire to have a maximum or
minimum bias for a pico-$p$ (say, $b_{p,max}$ and $b_{p,min}$). For example, if a pico is
deployed to fill a coverage hole or high demand area, then the $b_{min}$ should be such that UEs around
the coverage hole get associated with the pico. This can be handled using the
following steps: 

\begin{enumerate}

\item  First run the joint UE-association and ABS-determination algorithm
(Algorithm~\ref{algo:abs-relax} and Algorithm~\ref{algo:abs-round}) by setting
$r_u^{macro}=0$ for all UEs that get associated with the pico even with minimum
bias, and $r_u^{pico}=0$ for all UEs that do not get associated with the pico
even with maximum bias. 

\item Next execute the 3~steps for cell-bias determination described in this section
but with the minor modification in~(\ref{eqn:csbmmse}) so that the $\arg\min$
operation is restricted to $b\in[b_{p,min},b_{p,max}]$.

\end{enumerate}
\end{remark}

\vspace{-0.15in}
\subsection{Converting ABS numbers into ABS patterns} 

In the previous sections, we have
provided techniques to compute number of ABS subframes for every $N_{sf}$ subframes
(i.e., ABS subframes per ABS-period). In practice, we also need to specify the exact
subframes in an ABS-period that are used as ABS subframes. The ABS number can be
converted into a pattern as follows: 

\begin{enumerate}

\item Index the subframes in an ABS-period in a {\em consistent} manner across all macros
and picos.

\item Suppose a macro-$m$ leaves out $k$ out of $N_{sf}$ subframes as ABS subframes.
Then macro-$m$ offers the {\em first-$k$} subframes as ABS subframes 
where the {\em first-$k$}  relates to the indexing in the previous step.

\end{enumerate}

 This simple scheme works provided all macros have the same set  of permissible
sub-frames if required (i.e., there is no restriction on certain macro that it cannot
offer certain subframes as ABS).  Notice that, since a pico can effectively use the
least number of ABS offered by interfering macros, this scheme would naturally ensure
a provably correct mapping between number of ABS and ABS-pattern.

\vspace{-0.15in}
\section{Evaluation using RF Plan from a Real Network }
\label{sec:eval}

\changeclr
We evaluated our algorithms using RF plan from a real network deployment by a 
popular operator in Manhattan, New York City.
The goal of our evaluation is four folds. First, to compare
our proposed eICIC algorithm with other alternative schemes. Second, to understand
the optimality gap of our algorithm because we have shown that optimizing eICIC
parameters is NP-hard. Third, to understand the benefits offered by eICIC because
operators are still debating whether the additional complexity of eICIC is worth the
gains. Fourth, we show some preliminary results on how eICIC gains vary with pico
transmit power and UE density.

\normalclr

\vspace{-0.15in}
\subsection{Evaluation Framework}

{\bf Topology:} \changeclr We used an operational LTE network deployment by a leading operator
in New York City to generate signal propagation maps by plugging in the tower
and terrain information along with drive-test data into a commercially available RF
tool that is used by operators for cellular planing~\cite{9955tool}.  In
Figure~\ref{fig:nyctop}, we show the propagation map of the part of the city that we used for
evaluation, along with the macro-pico interference graph for nominal pico transmit
power of 4W. \normalclr The RF plan provides path loss estimates from actual macro location to
different parts of the city.  For the purpose of this study, we selected an area of
around 8.9~$\text{km}^2$ in the central business district of the city. This part of
the city has a very high density of macro eNB's due to high volume of mobile
data-traffic.  The macro eNB's are shown in blue color with sectorized antennas and
these eNB's are currently operational.  While macro cells used in our evaluation are
from the existing network, LTE pico cells are yet to be deployed in reality. Thus the
pico locations were manually embedded into the network planning tool. We carefully
chose 10~challenging locations for our picos: some are chosen with locations with
poor macro signals, some pico locations are chosen with high density of interfering
macros, some are chosen to coincide with traffic intensity hotspots, and one pico is
also deeply embedded into a macro-cell. The picos are shown in red circle with
omnidirectional antennas. 

All eNB's support $2\times 2$ MIMO transmissions. UEs have $2\times 2$ MIMO MMSE-receivers
that are also equipped with {\em interference cancellation} (IC) capability to cancel out broadcast
signals from macro during pico downlink transmissions over ABS subframes.

\begin{figure}[t]
\begin{center}
\subfigure{
\includegraphics[height=1.4in,width=1.6in]{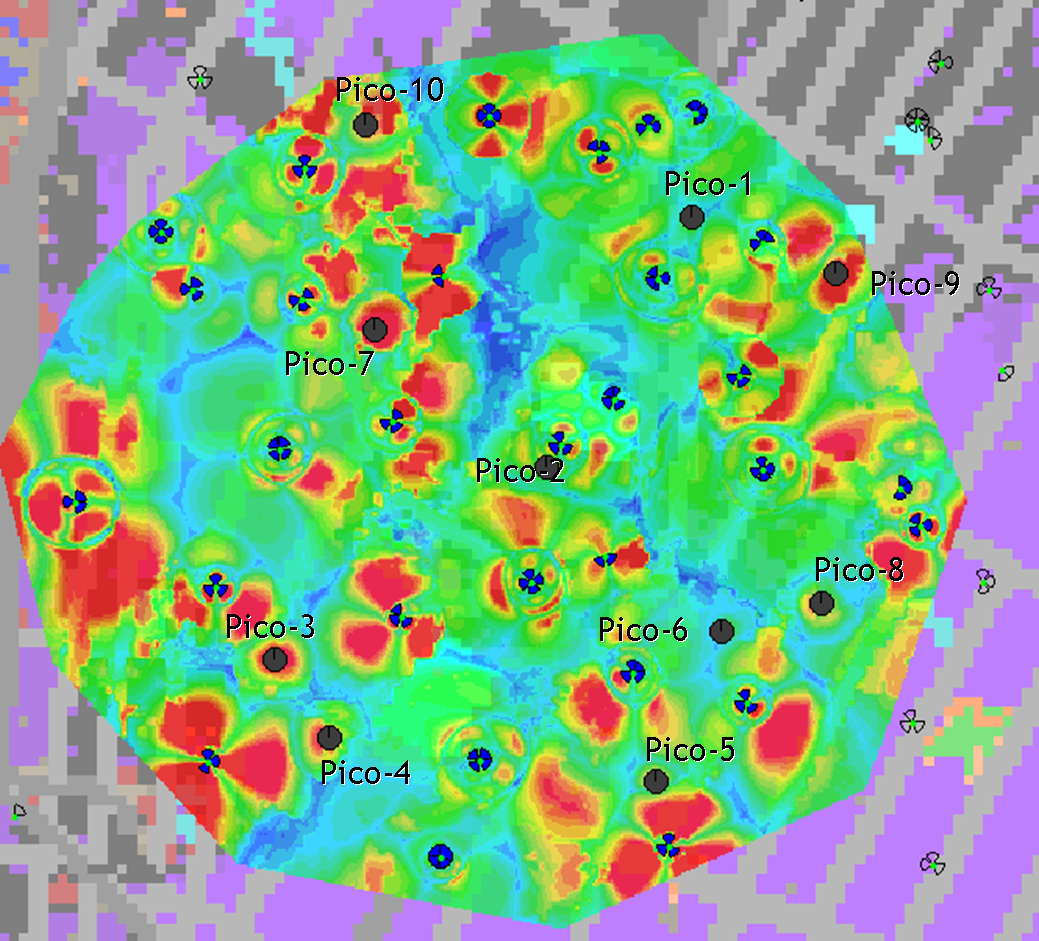}
}
\subfigure{
\includegraphics[height=1.4in,width=1.6in]{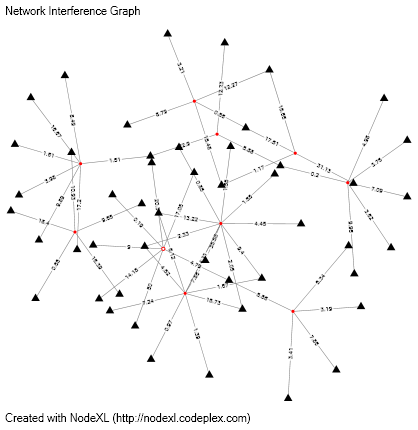}
}
\vspace{-0.1in}
\caption{\label{fig:nyctop}
Propagation map of the evaluated LTE network in New York City and the associated macro-pico interference graph.
The blue eNB's are the macros and are currently
operational. The grey eNB's are low power picos and we manually placed them in the
tool. Pico-10, Pico-3, Pico-5, Pico-9 have traffic hotspots around them.}
\end{center}
\vspace{-0.3in}
\end{figure}

{\bf Important Cell Parameters:} \changeclr The macro eNB's have transmit power of 45~dBm (31~W). For the picos,
we evaluated with 3 different settings of transmit powers: 36~dBm (4~W), 30~dBm (1~W), and
27~dBm (500~mW). \normalclr The bandwidth is 10~MHz in the
700~MHz LTE band. The pico heights are chosen as typically 30~ft above the ground
and macro heights are variable based on actual deployment and are typically much higher
(more than 100~ft in many instances.). We choose $N_{sf}=40$ so that we can obtain
the number of ABS offerings in every 40~subframes. Also, we allow a maximum bias of 15~dB for 
any pico because this is a typical restriction in current networks. 

\changeclr {\bf Traffic:} While the macros and the propagation map used in our
evaluation is for a real network, we create synthetic UE locations for our evaluation
because LTE pico deployments are still not very prevalent. This is done as follows.
In the area under consideration, we chose a nominal UE density of around
450~active~UEs/sq-km (dense urban density).  In addition, we created UE hotspots
around Pico-10, Pico-3, Pico-5, and Pico-9. The hotspots around Pico-3, Pico-5,
Pico-9 have double the nominal UE density and the traffic hotspot around Pico-10 has
50\% more UE density than nominal. We also performed evaluation by varying the UE
density around the macro cells to 225~active~UEs/sq-km (urban density) and
125~UEs/sq-km (sub-urban density) without altering the hotspot UE densities around
the selected picos. As we discuss in Section~\ref{sec:son}, in practice, network
measurements would be available in terms of average traffic load and SINR
distribution from which the UE locations can be sampled.

{\bf Methodology:} The radio network planning tool (RNP)~\cite{9955tool} and our
eICIC implementation were used to generate the results as follows.


1) RNP tool was used to generate signal propagation matrix in every pixel in the area of interest
in New York City as shown in Figure~\ref{fig:nyctop}.

2) The RNP tool was then used to drop thousands of UEs in
 several locations based on the aforementioned UE density profiles. All UEs had unit
weights (as in weighted proportional fair).

3) Based on the signal propagation matrix and UE locations, we invoked the
built-in simulation capability of RNP tool to generate the macro-pico interference
graph and the following donwlink SINR's for every UE: best-macro SINR, best-pico ABS
SINR, best-pico non-ABS SINR. These SINR values were converted to physical layer
rates using LTE look-up table. This step essentially produces the complete set of inputs for
OPT-ABS problem as follows.  

4) Then this input-data was fed into our implementation of proposed eICIC
and other comparative schemes described in Section~\ref{sec:otherschemes}.


Thus we used RNP too to generate synthetic input that is
representative of SINR and path-losses in a live network.

\vspace{-0.15in}
\subsection{Comparative eICIC Schemes}
\label{sec:otherschemes}

The three schemes we compare are as follows.


1) {\bf Proposed eICIC:} This is the proposed algorithm developed in this paper.
Just to summarize, we first apply Algorithm~\ref{algo:abs-relax} and the
rounding scheme in Algorithm~\ref{algo:abs-round} and finally we use the technique
described in Section~\ref{sec:csb} for obtaining CSB's.

2) {\bf Fixed eICIC Pattern:} Another option is to use a fixed or uniform eICIC pattern
across the entire network. In~\cite{perf-eicic12}, the authors have performed
evaluation with fixed eICIC patterns. Also~\cite{sim-eicic} considers fixed eICIC
parameters. We also compare our proposed eICIC algorithm to the following four (ABS,
CSB) combinations: (5/40, 5~dB), (10/40, 7.5~dB), (15/40, 10~dB), (15/40, 15~dB). The
fixed patterns represent the range of eICIC parameters 
considered in the literature.

3) {\bf Local Optimal Heuristic:} This is a local optimal based heuristic that is
very easy to implement and is also amenable to distributed implementation. This
scheme works as follows. First, each pico sets individual biases to maximize the total
improvement (as compared to zero-bias) of physical layer rates (by considering the
ABS rates) of all UEs within the coverage range of the pico. This step readily
provides the set of UEs that associate with picos. In the next step, each macro $m$
obtains the fraction (say, $a_{m}$) of UEs within its coverage range that
associates with itself and then the macro offers $\lceil N_{sf}(1-a_{m})\rceil$
as ABS sub-frames. Each pico can only use minimum number of ABS sub-frames offered by
its interfering macros.


\vspace{-0.15in}
\subsection{Results} 
For our results, we consider all UEs in the coverage area of deployed picos and all macros 
that interfere with any of these picos. Clearly, these are the only UEs that are affected by
eICIC or picos.

\begin{figure}[t]
\begin{center}
\includegraphics[height=1.25in,width=1.80in]{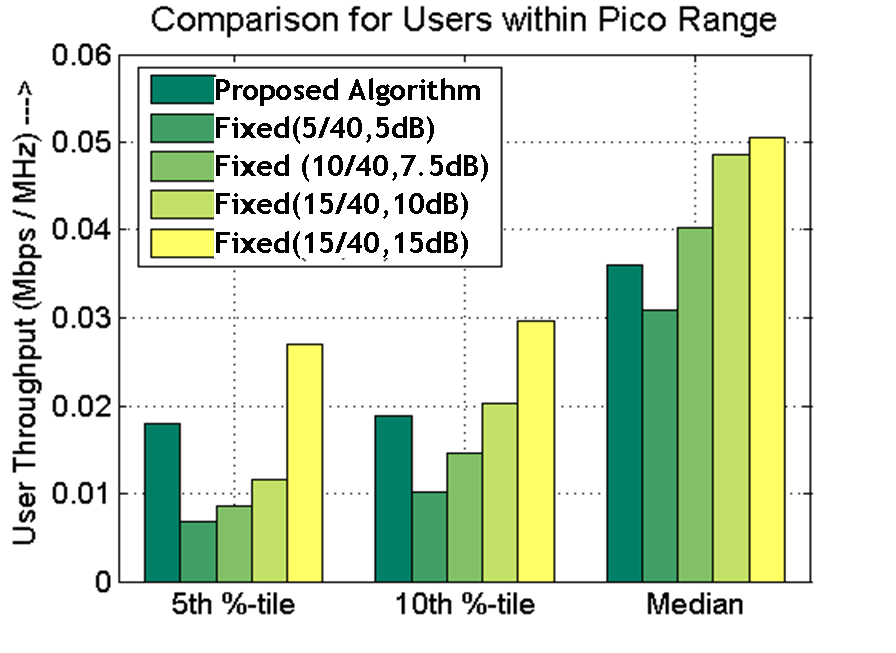}
\hspace{-0.20in}
\includegraphics[height=1.25in,width=1.80in]{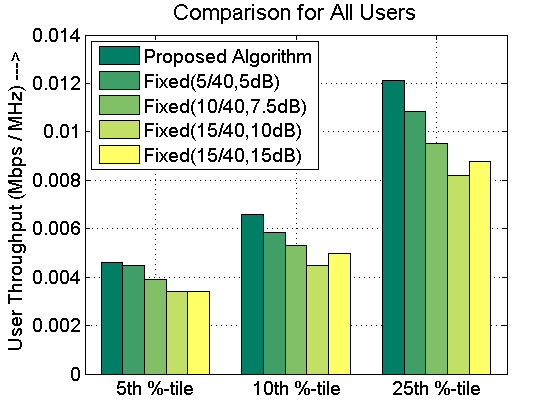}
\end{center}
\vspace{-0.15in}
\caption{\label{fig:compfixed}
\changeclr
Comparison of proposed eICIC with fixed eICIC pattern for
$5^{th}, 10^{th}, 50^{th}$ percentile of UE-throughput. Plots
are for pico transmit power of 4~W.\normalclr}
\vspace{-0.15in}
\end{figure}

\begin{figure}[t]
\begin{center}
\includegraphics[height=1.25in,width=1.80in]{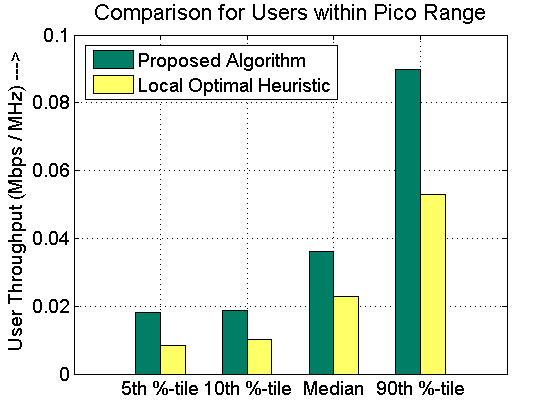}
\hspace{-0.20in}
\includegraphics[height=1.25in,width=1.80in]{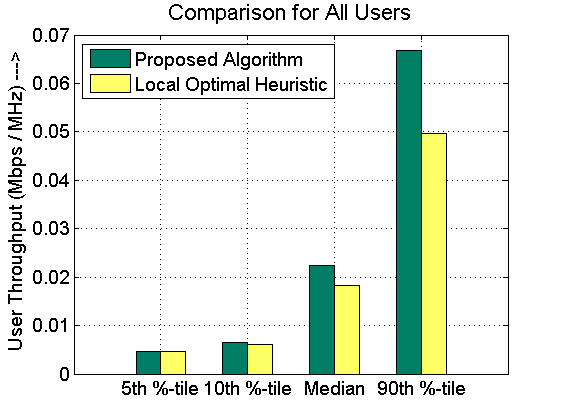}
\end{center}
\vspace{-0.15in}
\caption{\label{fig:complocal}
\changeclr Comparison of proposed eICIC with the local optimal heuristic for
$5^{th}, 10^{th}, 50^{th}, 90^{th}$ percentile of UE-throughput. The plots
are with pico transmit power of 4~W. \normalclr}
\vspace{-0.25in}
\end{figure}

{\bf Comparison with other schemes:} In Figure~\ref{fig:compfixed}, we compare our
algorithm to different network wide fixed ABS settings. The interesting comparison is
between our proposed eICIC algorithm and fixed (ABS, CSB) setting of (15/40, 15~dB)
which corresponds to a CSB value of maximum possible 15~dB; all other fixed eICIC
schemes perform poorly. This fixed eICIC setting of  (15/40, 15~dB) appears to
perform better than our scheme for all UEs in the pico footprint area because it
associates all UEs in the pico footprint area to the pico; whereas, our scheme does
not necessarily associate all UEs in the pico footprint area to the pico. However,
the fixed eICIC scheme fails to account for the overall network performance as the
macro UEs have to sacrifice far greater (compared to our scheme) throughput due to
eICIC.  Indeed, the throughput percentiles of all UEs in the system, for any fixed
eICIC scheme, is reduced compared to our algorithm as we can see from the plot in the
bottom panel of Fugure~\ref{fig:compfixed}. For example, our proposed eICIC improves
the $5^{th}, 10^{th}, 25^{th}$ percentile throughput of Fixed-(15/40, 15~dB) eICIC
configuration by $30\--40\%$.  For specific macros that do not interfere with any
hotspot picos, this improvement is more than 50\%.  In typical deployments where
many macros may have few or no pico neighbors in the macro-pico interference graph
unlike our evaluation topology, the loss of the overall system performance could be
more pronounced due to fixed eICIC configuration. Also, finding a good but fixed
eICIC setting could also be challenging. Table~\ref{tab:computil} also shows that the
overall log-utilty of the system is better with our proposed scheme compared to the
fixed eICIC schemes.

In Figure~\ref{fig:complocal}, we compare the proposed the eICIC with the local
optimal heuristic described in Section~\ref{sec:otherschemes}. Our scheme outperforms
the local optimal heuristic by a margin of more than $80\%$ for UEs in
pico-footprint area; furthermore, the overall systems performance is better with our
scheme as can be seen from the plot in the bottom panel of Figure~\ref{fig:complocal}
and Table~\ref{tab:computil}. However, the local optimal heuristic is very easy to
implement and could be promising with additional minor changes. We leave this as a
future topic of research.

\renewcommand{\tabcolsep}{5pt}
\begin{table}[t] \footnotesize
    \caption{\changeclr Comparison of total log-utility (total of logarithm
of UE throughputs in kbps/MHz) for different macro UE density and 
pico powers. {\bf DU, U, SU} stand for dense urban, urban, and sub-urban
UE density.\normalclr}
    \label{tab:computil}
\vspace{-0.2in}
\begin{center}
\begin{tabular}{| c | c | c | c | c | c | c |}
    \hline
(Macro  &  {Proposed} & {Local} & 
{ Fixed} & { Fixed} & { Fixed} &
{ Fixed}\\
Density, & eICIC & Opt & (5,5) & (10,7.5) & (15,10) & (15,15) \\
Pico  &&&&&& \\ 
Power) &&&&&&
\\ \hline
DU,4W &
5123.4 &  4799.4  & 4941.5 &  4886.3 &  4770.0 &  4837.1  
\\ \hline
DU,1W &
4984.0 &  4669.5  & 4786.1 &  4724.2 &  4609.8 &  4707.7  
\\ \hline
DU,$\smfrac{1}{2}$W &
4232.3 &  4018.6  & 4036.6 &  3976.7 &  3879.3 &  4001.3 
\\ \hline
U,4W &
3356.9 &  3154.2  & 3257.9 &  3227.9 &  3175.1 &  3212.6
\\ \hline
SU,4W &
2209.2 &  2032.9  & 2137.5 &  2124.5 &  2094.0 &  2123.8
\\ \hline
\end{tabular}
\end{center}
\vspace{-0.20in}
\end{table}

\normalclr
{\bf Optimality gap of our algorithm:} Since the solution to RELAXED-ABS is an upper
bound to the optimal solution of OPT-ABS, we obtain the optimality gap by comparing
our final solution to that produced by RELAXED-ABS ( Algorithm~\ref{algo:abs-relax}).
We compute $g$, such that our algorithm is within $100\times(1-g)\%$
of the optimal, as follows. Suppose
$R_u^{rel}$ and $R_u^{alg}$ be the UE-$u$'s throughputs produced by RELAXED-ABS and
our complete algorithm respectively. Then, we say that the optimality-gap is a factor
$g<1$ if $ \sum_u \ln(R_u^{alg}) \geq \sum_u \ln(R_u^{rel}(1-g)) $.  The smallest
value of $g$ that satisfies this can easily be computed.  \changeclr In Table~\ref{tab:compopt},
we show for various settings of macro UE densities and pico transmission power
that, our scheme is typically within $90\%$ of the optimal. \normalclr
\begin{table}[h] \footnotesize
    \caption{\changeclr Optimality gap of our algorithm
for different macro UE density and 
pico powers. {\bf DU, U, SU} stand for dense urban, urban, and sub-urban
UE density. \normalclr}
    \label{tab:compopt}
\vspace{-0.2in}
\begin{center}
\begin{tabular}{| c || c | c | c | c | c |}
\hline
(Macro density, & DU, 4W & DU, 1W & DU,$\smfrac{1}{2}$W & U, 4W & SU, 4W \\
Pico power) & & & & & 
\\ \hline
\% of Optimal & 93.77\% &  95.64\% & 95.86\% & 92.98\% & 97.03\% \\
\hline
\end{tabular}
\end{center}
\vspace{-0.10in}
\end{table}

\changeclr
{\bf Benefits of eICIC:}  In the
interference graph, we have 26 macros and 10 picos. In typical deployments, there are
going to be many more picos and macros. Thus, to understand the gains that even a few
picos can offer, we show the following plots in Figure~\ref{fig:compcdf}: CDF of
throughputs of UEs in the pico coverage area, and CDF of throughput of UEs
outside of pico coverage area. Thus, we wish to understand the gains of UE who
could potentially associate with the picos, and the performance impact of UEs who
do not have the option of associating with picos. The plot in the top panel of
Figure~\ref{fig:compcdf} shows the throughput gains: {\em (i)} compared to no eICIC
based scheme the gains are more than 200\% for the far-edge UEs (say, $2.5^{th}$
percentile of the throughputs) and  40\--55\% for edge UEs ($5^{th}\--10^{th}$
percentile of UE throughput), {\em (ii)} also, compared to no pico, the gains are
even more dramatic and around $300\%$ even for  $5^{th}$ percentile of the
throughputs. The plot in the bottom panel of Figure~\ref{fig:compcdf} shows that the
throughput gains (over no eICIC based pico deployment) of pico UEs do not come at
an appreciable expense of macro UEs' throughput. In other words, though the macro
eNB's have fewer subframes for transmissions (due to ABS offered to picos) using
eICIC, this is compensated by the fact the macro UEs compete with fewer UEs (many
UEs end up associating with picos under eICIC). Thus, there are great benefits of
not only pico deployments, but also eICIC based pico deployments.
\normalclr

\begin{figure}[t]
\begin{center}
\includegraphics[height=1.25in,width=1.80in]{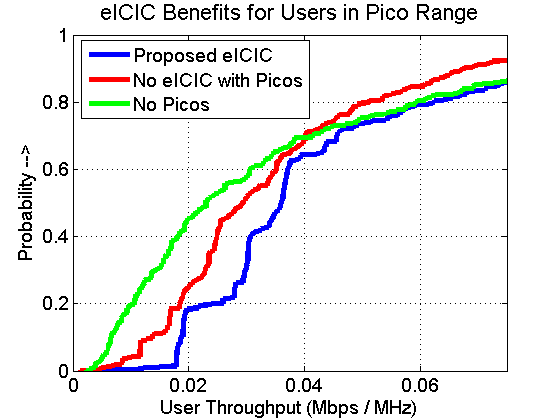}
\hspace{-0.20in}
\includegraphics[height=1.25in,width=1.80in]{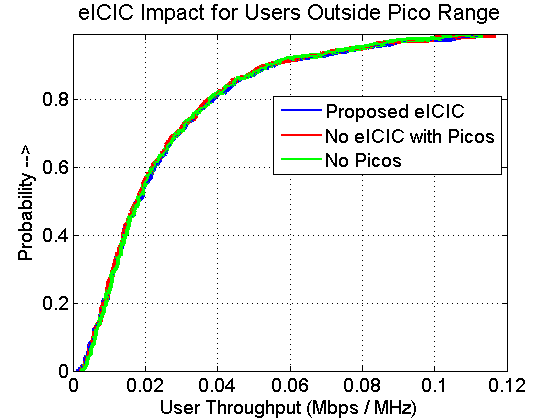}
\end{center}
\vspace{-0.15in}
\caption{\label{fig:compcdf}
CDF of UE throughputs with out proposed eICIC, no eICIC with picos, and no picos.
The plots are with pico transmit power of 4~W and dense urban macro density.}
\vspace{-0.15in}
\end{figure}

\changeclr
{\bf eICIC gains with power and load variation:} To better understand the eICIC
gains, in Figure~\ref{fig:gainpowden} we compare the gains of eICIC using our
algorithm with pico deployment without eICIC by varying the pico transmit powers and
macro UE densities. In the top panel, we show the percentage throughput gain of
eICIC scheme for different pico transmit powers for $5^{th}, 10^{th}, 50^{th}$
percentile of UE throughputs. It can be seen that, it is the edge UEs who really
gain with eICIC; indeed, this gain could even come at the expense of UEs close to
the pico (as can be seen with 1~W pico power scenario) who do not gain much due to
eICIC. The edge gain is also a direct consequence of our choice of log-utility
function as system utility. In the bottom panel of Figure~\ref{fig:gainpowden}, we
show the gains for different macro UE density. It can be see that, higher macro
UE density results in higher gain due to eICIC. Intuitively speaking, {\em more the
UEs that have the choice of associating with picos, larger are the eICIC gains from
our algorithm.} This suggests the usefulness of our scheme for practical scenarios
with large number of picos and very high density in the traffic hot-spot areas.
\normalclr

\begin{figure}[t]
\begin{center}
\includegraphics[height=1.25in,width=1.80in]{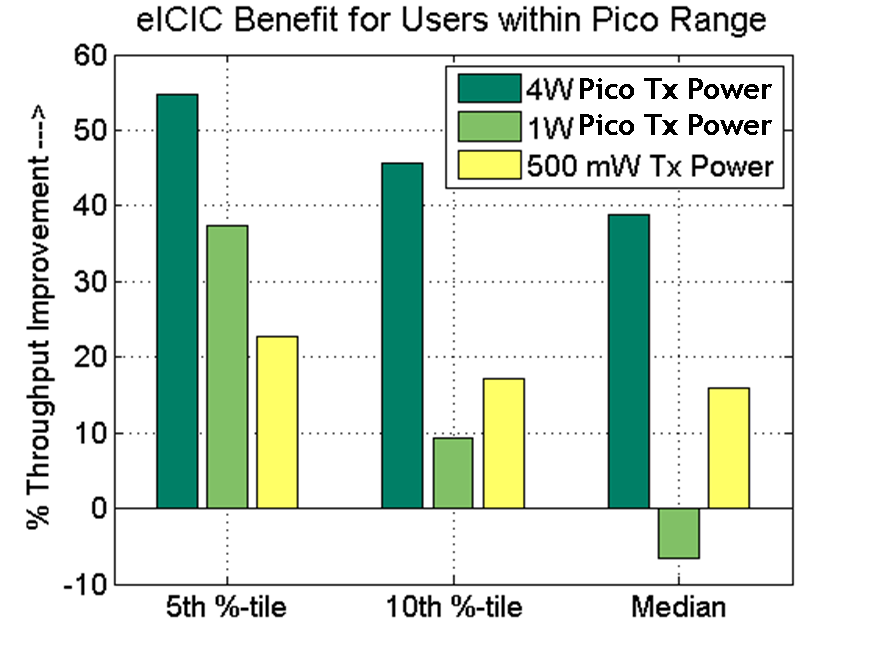}
\hspace{-0.20in}
\includegraphics[height=1.25in,width=1.80in]{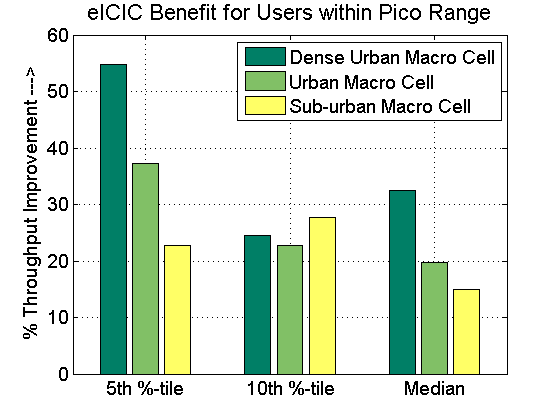}
\end{center}
\vspace{-0.15in}
\caption{\label{fig:gainpowden}
\changeclr
\%-Improvement in $5^{th}, 10^{th}, 50^{th}$ percentile of UE throughput with
eICIC as opposed to no eICIC based pico deployment for different pico transmits
powers (left-panel plot) and different macro cell densities (right-panel plot).
\normalclr}
\vspace{-0.15in}
\end{figure}

\renewcommand{\tabcolsep}{4pt}
\begin{table}[t]\footnotesize
    \caption{Representative ABS and CSB Values (using IC receivers)}
    \label{tab:abscsb}
\vspace{-0.2in}
\begin{center}
{\footnotesize
\begin{tabular}{| c ||l|l|l|l|l|l|l|l|l|l| }
\hline
Pico &1 &2 &3 &4 &5 &6 &7 &8& 9& 10 \\
Index &&&&&&&&&& \\
\hline
ABS &6 &4 &5 &5 &4 &4 &3 &4 &7 &7\\ 
(out &&&&&&&&&& \\
of 40) &&&&&&&&&& \\
\hline
CSB &13.7 &8.4 &13.4 &13.7 &7 &9.4 &13.9 &14 &7.4 &7.8\\ 
(dB) &&&&&&&&&& \\
\hline
\end{tabular}
}
\end{center}
\vspace{-0.20in}
\end{table}
\renewcommand{\tabcolsep}{6pt}

{\bf Optimal Parameters:} In Table~\ref{tab:abscsb}, we show the optimal ABS received
by each pico and the associated bias obtained using our algorithm. There are a couple
of interesting observations. First, the ABS offered to picos not only depends on
traffic load but also depends on the number of interferers. For example, Pico-5
received 4 out 40 subframes for ABS though it has a hotspot around it, however,
Pico-1 received 6 out of 40 subframes for ABS without any hotspot around it. This is
because, Pico-5 has more neighbors in the macro-pico interference graph. Second,
Pico-2 which is embedded into a macro, also receives 4~ABS subframes and serves as
enhancing in-cell throughput. Thus picos can go beyond improving throughput in edges
if eICIC parameters are configured in a suitable manner. This also shows that there
could be considerable variation in optimal ABS and CSB settings. This explains
the poor performance of network wide fixed eICIC schemes.

\vspace{-0.15in}
\changeclr
\section{SON and eICIC: Challenges and Discussion}
\label{sec:son}


A key aspect of LTE networks is its {\em Self Optimized Networking} (SON) capability. Thus, it
is imperative to establish a SON based approach to eICIC parameter
configuration of an LTE network. The main algorithmic computations of SON
may be implemented in an centralized or a distributed manner. In the
centralized computation, the intelligence is concentrated at the Operations Support
System (OSS) layer of the network, while in the distributed computation the
computation happens in the RAN or eNB. The main benefit of a centralized approach
over a distributed approach is twofold: a centralized solution in OSS is capable of
working across base stations from different vendors as is typically the case, and
well-engineered centralized solutions do not suffer from convergence issues of
distributed schemes (due to asynchrony and message latency). Indeed, realizing these
benefits, some operators have already started deploying centralized SON for their cellular networks.
Nevertheless, both centralized and distributed approaches
have their merits and demerits depending on the use-case.  Also, it is widely
accepted that, even if the key algorithmic computations happen centrally in OSS, an
overall hybrid architecture (where most heavy-duty computations happen centrally in
OSS with distributed monitoring assistance from RAN) is best suited for  complicated SON
use-cases such as eICIC whereas complete distributed approach is suited for simple use-cases
like cell-neighbor detection. We next discuss our prototype hybrid SON for eICIC.

\vspace{-0.15in}
\subsection{Hybrid SON Architecture}

The architecture of our prototype is shown in Figure~\ref{fig:impl}. Apart from an operational 
wireless network, the architecture has two major component-blocks:  a network planning tool and an engine for
computing optimal eICIC configuration.

In our prototype the main optimization task is executed centrally at
the OSS level,  but input distributions are provided by the RAN and further
estimated/processed in the OSS. From the perspective of the optimization algorithm,
centralization offers the best possible globally optimal solution given accurate and
timely data inputs. We also assume that other affiliated SON procedures, such as
Automatic Neighbor Relations (ANR) are executed in the RAN and their results reported
to the OSS. After the execution of the OSS optimization algorithms, the optimal
parameters propagate southbound towards the RAN elements.

{\bf Monitoring in the RAN:} 
As part of the Operations Administration and Maintenance (OA\&M) interfaces,
Performance Management (PM) data are reported to all OSS applications,
SON-applications including. In general, we can have periodic reporting or event-based
reporting, an implementation choice that trades latency and accuracy. Irrespective of
the implementation, eICIC requires from the RAN, path loss statistics, traffic load
statistics and SINR statistics. Our prototype implementation can flexibly accommodate
the case of  missing data i.e. incomplete statistics, that is common during the
planning phase of the network or in the case where the required inputs are not
readily available by another vendor's RAN implementation.  In these cases we can
easily replace actual network data with synthetic data generated by a radio network
planning tool. 

{\bf The role of radio network planning (RNP) tools:} Various Databases are used to
import information necessary for performing the radio network planning. Inventory
information that provides network topology, drive-tests that calibrate path loss
models as well as performance measurement data that determine the shape and value of
traffic intensity polygons in the area of interest, are the most important
information sources aggregated in the RNP tool. RNP can use this information to
generate synthetic input data (using built-in simulation capabilities of the tool)
for our eICIC algorithm. For our purpose, we used a planning tool~\cite{9955tool}
that uses the traffic map, propagation map, and eNodeB locations to generate multiple
snapshots of UE locations.  This input data is saved into a database.  The eICIC
computation engine implements our proposed algorithm to compute optimal eICIC
parameters.  The role of the radio network planning tool as a prior distribution
generator is quite important in the planning as well as the initialization procedures
of a network element when measurements are unavailable or too sparse.    

\begin{figure}[t]
\begin{center}
\includegraphics[height=1.6in,width=2.8in]{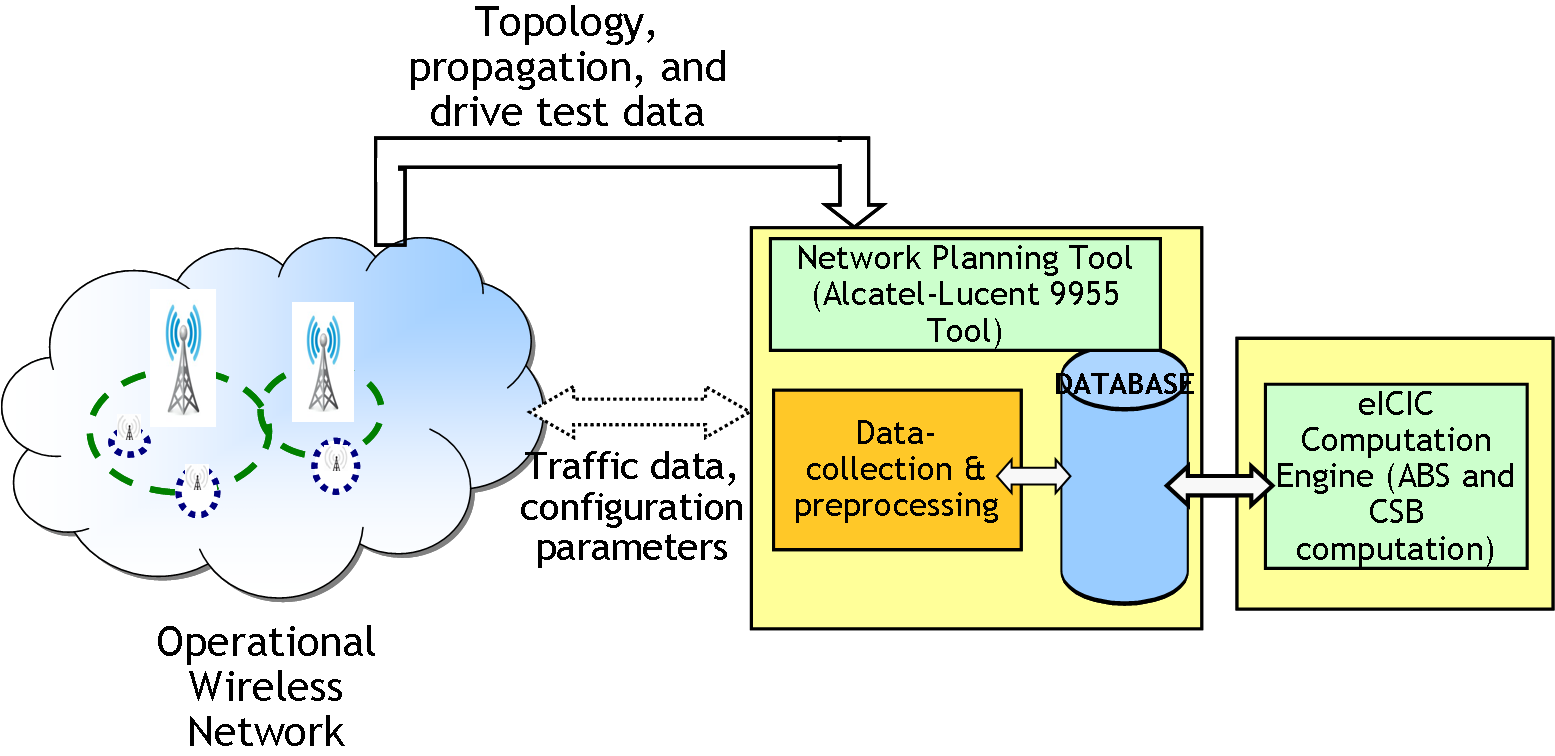}
\caption{\label{fig:impl}
Prototype Architecture.}
\end{center}
\vspace{-0.33in}
\end{figure}

\vspace{-0.15in}
\subsection{Computational Flow}  

{\bf Average Load Based Input:}
The traffic data that is available from the real network comprises of average traffic load information 
over a period of interest. Our method can be used for the purpose of obtaining the correct ABS and CSB
configuration can be used as follows.

We consider a scenario where, for each cell, average traffic load and SINR distribution
(in different subframes including ABS subframes) are reported periodically. Indeed,
such reporting is common in many real network deployments~\cite{ltehandbook}.  
We adapt our solution as follows:

\begin{enumerate}

\item {\em Generation of sample UE location:} Based on the traffic map and SINR
distribution, multiple system wide sample UE-location snapshots are generated
(commercial network planning tools usually have this capability). Each sample
UE-locations are translated into downlink PHY-layer rate between UE and Macro, UE and
Pico with and without ABS.  Note that, the exact location of a sample UE is of little
relevance, rather, the SINR to the nearest macro and pico along with the RSRP values
are of relevance here.

\item {\em Optimization for each snapshot:} Based on the UE-location snapshots, for
each such sample snapshot, our algorithm is run for solving OPT-ABS. This generates
ABS and CSB configuration for each sample.

\item {\em Monte-Carlo averaging:} Once ABS and CSB computation is performed for
sufficient sample UE location snapshots, the results are averaged over all samples. 

\end{enumerate}

{\bf Time-scale and dynamic eICIC:} There are two important considerations for the
reconfiguration frequency of eICIC parameters. Firstly, since only the traffic
distribution statistics can be obtained from the network, it is imperative that eICIC
computations happen at the same time-scale at which the traffic distribution can be
estimated accurately; otherwise, eICIC changes at a faster time-scale may 
not provide appreciable gains while causing unnecessary reconfiguration overhead.
Secondly, since eICIC reconfiguration involves a cluster of
macros and picos, it takes a few
minutes to have new set of traffic information from all the cells~\cite{ltehandbook}.
Therefore, the time scales of changing eICIC parameters for all practical purposes
are in the order of few minutes (typically, $5\--15$ minutes). Thus, eICIC
reconfiguration ought to happen whenever {\em (i)} traffic load changes significantly
in some cells, or {\em (ii)} or when a maximum duration elapses since the last
reconfiguration. In addition, if the estimated improvement upon new eICIC re-computation 
is small, then the previous eICIC configuration can be maintained.


\normalclr
  
\vspace{-0.15in}
\subsection{Fully Distributed SON Architecture}
\label{sec:dist}

Another implementation option is to distribute eICC computations at the Network Element
(eNB) level and evidently what is traded off with this approach, as compared to
centralized, is signaling and communication latency. \changeclr In LTE, the X2 interface can be used
to interconnect eNBs and this interface has been the subject of extensive standardization
when it comes to interference management. The main challenge in distributed approach
comes from X2 latency and asynchrony. \normalclr

Using proprietary messages over X2, our solution can be adapted for in-network
computation where the macro and picos use their local computational resources to
obtain a desired solution. We provide a broad outline in the following. The key to
achieving this is our dual based implementation of ABS-RELAX.
Suppose the optimization task has to be performed periodically or over a time-window
of interest. Then, a distributed implementation involves the following high-level
steps. {\em (i)} Macros and picos generate UE samples in the cells based on
traffic/SINR distribution of UEs in the time-window of interest; the sample data is
exchanged between neighboring macro-pico pairs. {\em (ii)}
Algorithm~\ref{algo:abs-relax} is run in a distributed manner, where, in each
iteration, macros and picos exchange dual variables $\mu_{p,m}$ and $\lambda_u$'s
for relevant UEs. This exchange allows macros and picos to update primal variables
(for themselves and also candidate UEs) locally. {\em (iii)} Finally, note that the
rounding step and UE association in Algorithm~\ref{algo:abs-round} can be performed
locally. If UE association is implemented via cell-selection bias, the scheme
described in Section~\ref{sec:csb} can be easily carried out locally at each pico
(this may require collecting UE association vector by message exchange with
neighboring macros).

\vspace{-0.12in}
\section{Concluding Remarks}

In this work, we have developed algorithms for optimal configuration of eICIC
parameters based on actual network topology, propagation data, traffic load etc. Our
results using a real RF plan demonstrates the huge gains that can be had using
such a joint optimization of ABS and UE-association based on real network data. The
broader implication of our work is that, to get the best out of wireless networks,
networks must be optimized based on real network data.

\vspace{-0.12in}

 \bibliographystyle{plain}
 \bibliography{myref}

\end{document}